\documentclass[twocolumn,prl,groupedaddress,showpacs,superscriptaddress]{revtex4-2}
\usepackage[utf8]{inputenc}

\usepackage{hyperref}
\usepackage{hyperref}
\usepackage{booktabs}
\usepackage{upgreek} %upright greek letters
\usepackage{float}
\usepackage{graphicx}
\usepackage{epstopdf}
\usepackage{tikz-cd}
\usepackage{quiver}
\usepackage{amscd,amssymb,amsmath,latexsym,bm}
\usepackage{amsmath,amsthm,amssymb}
\usepackage{amsfonts}
\usepackage[mathcal,mathscr]{euscript}
\usepackage{braket}

\renewcommand{\vec}[1]{\boldsymbol{#1}}

\newtheorem{theorem}{Theorem}[section]

\newtheorem{corollary}[theorem]{Corollary}
\newtheorem{proposition}[theorem]{Proposition}
\newcommand{\CM}{{\mathbb C}}
\newcommand{\NM}{{\mathbb N}}
\newcommand{\QM}{{\mathbb Q}}

\newcommand{\ZM}{{\mathbb Z}}

\newcommand{\DM}{{\mathbb D}}

\newcommand{\Ff}{{\mathcal F}}

\newcommand{\Ll}{{\mathcal L}}

\begin{document}

\title{Converging Periodic Boundary Conditions and Detection of Topological Gaps on Regular Hyperbolic Tessellations }

\author{Fabian R. Lux}
\email{fabian.lux@yu.edu}
\affiliation{Department of Physics and
 Department of Mathematical Sciences 
\\Yeshiva University, New York, NY 10016, USA
}

\author{Emil Prodan}
\email{prodan@yu.edu}
\address{Department of Physics and
 Department of Mathematical Sciences 
\\Yeshiva University, New York, NY 10016, USA
}

\date{\today}

\begin{abstract} 
Tessellations of the hyperbolic spaces by regular polygons are becoming popular because they support discrete quantum and classical models displaying unique spectral and topological characteristics. Resolving the true bulk spectra and the thermodynamic response functions of these models requires converging periodic boundary conditions and our work delivers a practical solution for this open problem on generic $\{p,q\}$-tessellations. This enables us to identify the true spectral gaps of bulk Hamiltonians and, as an application, we construct all but one topological models that deliver the topological gaps predicted by the $K$-theory of the lattices. We demonstrate the emergence of the expected topological spectral flows whenever two such bulk models are deformed into each other and, additionally, we prove the emergence of topological channels whenever a soft physical interface is created between different topological classes of Hamiltonians.
\end{abstract}

\maketitle

The investigation of synthetic materials is an active area of research. In particular, crystals generated from tessellations of hyperbolic spaces have been proposed and some even realized with quantum, photonic, electromagnetic and mechanical degrees of freedom \cite{KollarNature2019,LenggenhagerNatComm2022,ZhangNatComm2022,ChenNatComm2023,RuzzeneEM2021,ZhangNatComm2023}. This is part of a new trend in materials science where the focus is shifted from making a material stronger, lighter, more durable, etc., to making it different or to behave differently. The new paradigm is geared towards creating new opportunities in materials science, which can come in the form of unique spectral characteristics or stabilization of fundamentally distinct topological phases and the hyperbolic crystals have been a source of both \cite{KollarCMP2020,ComtetAP1987,ComtetAP1987,ZhangNatComm2022,
ChenNatComm2023,ZhangNatComm2023,UrwylerPRL2022,LiuPRB2023,MathaiATMP2019}. In fact, the band topology of the hyperbolic crystals has been exhaustively characterized quite a while ago \cite{CareyCMP1998,MarcolliCCM1999,MarcolliCMP2001,LuekKTh2000}.

In these new endeavors, scientists are facing challenges that requires entirely new tools of analysis, both theoretical and computational, and lack of such tools can hold a field hostage for years. For example, the lack of a systematic way to impose periodic boundary conditions (PBC) on hyperbolic lattices prevents us from resolving the true bulk spectra of the Hamiltonians, computing thermodynamic coefficients, correlations functions, bulk topological invariants and identifying topological gaps. We recall that the hyperbolic space-groups do not have finite-index abelian subgroups, as their Euclidean counterparts do, and they are non-amenable, hence the ratio between the numbers of boundary and bulk sites converges to a strictly positive value. Thus, suppressing the boundary states in finite-size samples is necessary but, unfortunately, not sufficient because convergence with the sample size cannot be taken for granted. For these reasons, resolving the true bulk characteristics of hyperbolic crystals remained an open problem. 

If one is interested only in the bulk spectra, then a universal solution is to evaluate the local density of states at or near the center of the finite-size crystal with open boundary conditions \cite{ChenNatComm2023}, but this method converges only as an inverse power with the crystal size \cite{MassattMMS2017} and its reliability when it comes to computing thermodynamic coefficients or topological invariants is yet to be demonstrated. On the other hand, when PBCs can be systematically defined, they supply extremely fast convergences (typically exponential) with the size of the crystals, for both spectra and thermodynamic coefficients \cite{ProdanSpringer2017,LuxArxiv2023}. Partial progress on resolving the bulk characteristics of hyperbolic lattices has been achieved via generalizations of Bloch-Floquet calculus \cite{MaciejkoSciAdv2021,ChengPRL2022}, which is intimately related to the problem of PBCs \cite{MaciejkoPNAS2022}. So far, these techniques can resolve the bulk spectra covered by one and two dimensional representations of the hyperbolic space-groups and sometimes this seems to be just enough \cite{ChenNatComm2023,ZhangNatComm2023}.  

\begin{figure*}[t]
\center
\includegraphics[width=0.99\linewidth]{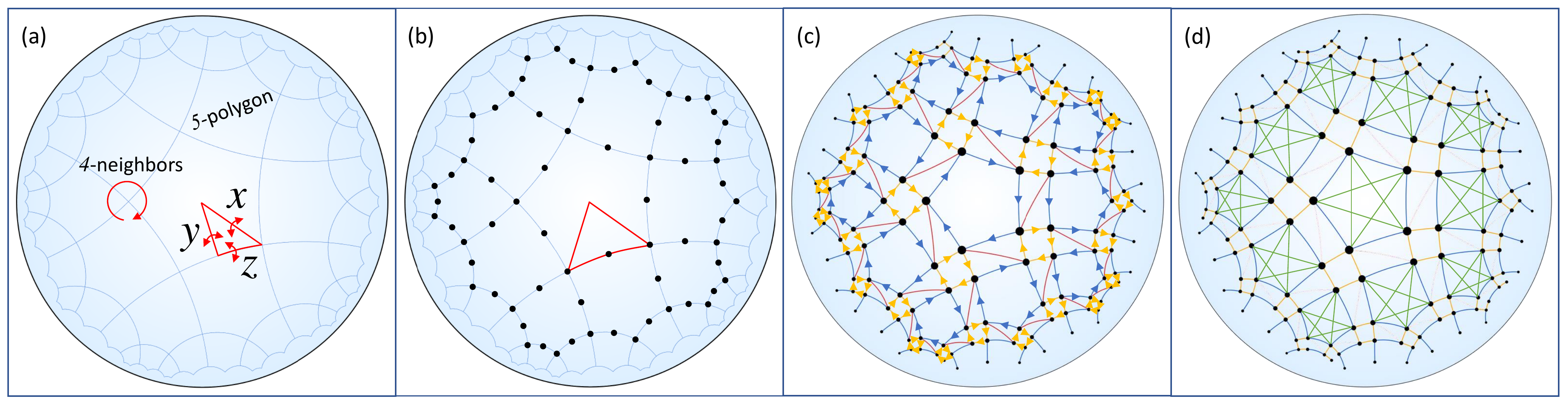}
  \caption{\small (a) The $\{5,4\}$ tilling of the hyperbolic disk and the reflections generating the full space-group of symmetries $\Delta_{\{p,q\}}$. (b) Vertices of the tilling and the fundamental domain of the proper space-group symmetries $\Delta^+_{\{p,q\}}$. (c) The Cayley diagraph of $\Delta^+_{\{p,q\}}$, showing the flow of the points under the right action of the generators $x_1$ (blue), $x_2$ (orange) and $x_3$ (red). (d) Couplings needed to implement $\pi_R(h_1(\lambda))$ from Eq.~\eqref{Eq:ModelH}, where ${\rm blue}=\lambda |g\cdot x_1^{-1}\rangle \langle g| + \lambda^{-1} |g \rangle \langle g \cdot x_1^{-1}|$, ${\rm green}=\lambda^2 |g\cdot x_1^{-2}\rangle \langle g| + \lambda^{-2} |g \rangle \langle g \cdot x_1^{-2}|$ and ${\rm orange}=|g\cdot x_2^{-1}\rangle \langle g| + |g \rangle \langle g \cdot x_2^{-1}|$. The latter is only needed for the integrity of the structure.
}
 \label{Fig:Tilling}
\end{figure*}

In our recent work \cite{LuxArxiv2023}, however, we introduced a general and algorithmic method to impose PBCs on increasingly larger finite hyperbolic crystals, together with rigorous and numerical proofs of fast convergence to the thermodynamic limit. The folding of the infinite lattice into a finite regular graph without boundary can be achieved by taking the quotient of the hyperbolic space-group with one of its finite-index normal subgroups \cite{LuekGFA1994,LuekBook2002,MaciejkoPNAS2022,SaussetJPA2007}. To converge to the thermodynamic limit, one needs a whole coherent sequence of such such normal subgroups, whose total intersection reduces to just the neutral element \cite{LuekGFA1994,LuekBook2002}. Note that, while a generic group can have a plethora of normal subgroups, only a coherent sequence can guarantee that the quotient groups are meaningful and actually resemble the infinite group that we want to approximate. This is a trivial task for regular lattices in Euclidean space, precisely because all space-groups contain the subgroup of pure translations, which is a finite-index abelian normal subgroup. This is not at all the case in the hyperbolic spaces yet the algorithm devised in \cite{LuxArxiv2023} does just that and, to work, it requires at the input a faithful representation of the hyperbolic space group as a subgroup of ${\rm GL}(n,R)$, where $R$ is a ring extension of $\ZM$. This input was provided in \cite{LuxArxiv2023} for the simplest hyperbolic space-group. In the present work, we supply this input together with computer algorithms that completely solve the problem of converging PBCs on generic regular $\{p,q\}$-tessellations of the hyperbolic disk. Their convergence was already demonstrated by testing against known exact results \cite{LuxArxiv2023} and will be tested again here.

As an application, we resolve all topological bands supported by a $\{p,q\}$-tessellation, except one \cite{Footnote1}, in the sense that we build gapped model Hamiltonians that display these bands at the bottom of the spectrum. Given that our numerically computed spectra are clean of boundary spectrum, we can demonstrate the main topological characteristic of these models, namely, the emergence of a topological spectral flow whenever a distinct pair of topological Hamiltonians are continuously deformed into each other. The topological bands are listed by the $K_0$-group of the space-group's $C^\ast$-algebra and we simply took this information from \cite{MarcolliCCM1999,LuekKTh2000}.

Let $\DM$ be the open disk model of the hyperbolic 2-dymensional space. Topologically, it is identical with the Euclidean disk, but $\DM$ carries the metric $d s^2 =\frac{dz d\bar z}{(1-|z|^2)^2}$. The homeorphisms of the disk preserving this metric form the continuous group ${\rm Iso}(\DM)$ of hyperbolic isometries. Its discrete subgroups of orientation preserving disk transformations with compact domain are called Fuchsian groups of first kind. Up to isometries, they are classified by their signatures $\langle g;\nu_1,\ldots,\nu_r \rangle$ and a Fuchsian group with such signature has $2g+r$ generators and can be presented using relations as \vspace{-0.2cm}
\begin{equation}\label{Eq:GenF}
\begin{aligned}
\Ff_{g,\bm \nu} = &  \big \langle a_1,b_1,\ldots,a_g,b_g,x_1,\ldots,x_m \ | \\ 
& \quad   x_1^{\nu_1}, \ldots, x_m^{\nu_m},  x_1 \cdots x_m [a_1,b_1] \cdots [a_g, b_g] \big \rangle, 
\end{aligned}\vspace{-0.2cm}
\end{equation}
where $[a,b] : =ab a^{-1} b^{-1}$ denotes the commutator of two elements. The generators $a_i$ and $b_i$ are called hyperbolic transformations and they do not fix any point of $\DM$, while the elements $x_i$ are called elliptic transformations and they do fix points of $\DM$ \cite[Ch.~2]{Katok1992}. Thus, there are similarities with the space-groups of the Euclidean spaces. In fact, any Fuchsian group has a (non-unique) normal subgroup $\Ff_{g'}$ generated by $2g'$ hyperbolic transformations, entering the exact sequence of groups \cite{MathaiATMP2019} \vspace{-0.2cm}
\begin{equation}\label{Eq:SES}
1 \rightarrow \Ff_{g'} \rightarrow \Ff_{g,{\bm \nu}} \rightarrow P=\Ff_{g,{\bm \nu}}/\Ff_{g'} \rightarrow 1, \vspace{-0.2cm}
\end{equation}
where $P$ is the (finite) point group of the hyperbolic crystal. The relation between Eq.~\eqref{Eq:SES} and the recent work \cite{BoettcherPRB2022} on hyperbolic crystallography is discussed in \cite{Supplemental}. Applications of the disk transformations contained in the Fuchsian groups of first kind to a point of the disk generates all discrete hyperbolic lattices.

The tessellations of $\DM$ by regular $\{p,q\}$-polygons are always possible if $1/p + 1/q < 1/2$. We will use $\{5,4\}$ in our exemplifications because it puts more points towards the center of disk, hence it is easier to visualize and fabricate. It is shown in Fig.~\ref{Fig:Tilling}(a). The full group of hyperbolic isometries preserving this tilling is the triangle group $\Delta_{\{5,4\}}$ generated by the three reflections $x$, $y$, $z$ against the sides of the triangle shown in Fig.~\ref{Fig:Tilling}(a) \cite{Katok1992}. It has a maximal subgroup of proper disk transformations, which is the Fuchsian group $\Delta^+_{\{p,q\}} = \langle 0;p,q,2\rangle$ with $x_1=xy$, $x_2=yz$ and $x_3=zx$, having the fundamental domain indicated in Fig.~\ref{Fig:Tilling}(b), which is strictly smaller than the regular polygon of the original tilling. Iff $q$ has a prime divisor less than or equal to $p$, then there exists a Fuchsian group having the regular polygon itself as its fundamental domain \cite{YunckenMMJ2003}. It is more fruitful, however, to proceed with $\Delta^+_{\{p,q\}}$ as the space-group of the regular tessellation.  

Tillings do not automatically come with vertices. However, Katok \cite[p.~70]{Katok1992} introduces a generic definition of vertices as the points of the tiles' boundaries that are fixed by one element of its full group of proper symmetries. The vertices generated by this criterion are shown in Fig.~\ref{Fig:Tilling}(b). In fact, this criterion generates a representation of the Cayley graph of the space-group $\Delta^+_{\{p,q\}}$. Other representations can be generated by acting with the elements of $\Delta^+_{\{p,q\}}$ on a generic point of the disk, as shown in Fig.~\ref{Fig:Tilling}(c). In fact, this figure shows the standard Cayley diagraph of $\Delta^+_{\{p,q\}}$ space-group, which encodes the entire group-algebraic information in a geometric fashion \cite{CoxeterBook}. 

\begin{figure}[t]
\center
\includegraphics[width=\linewidth]{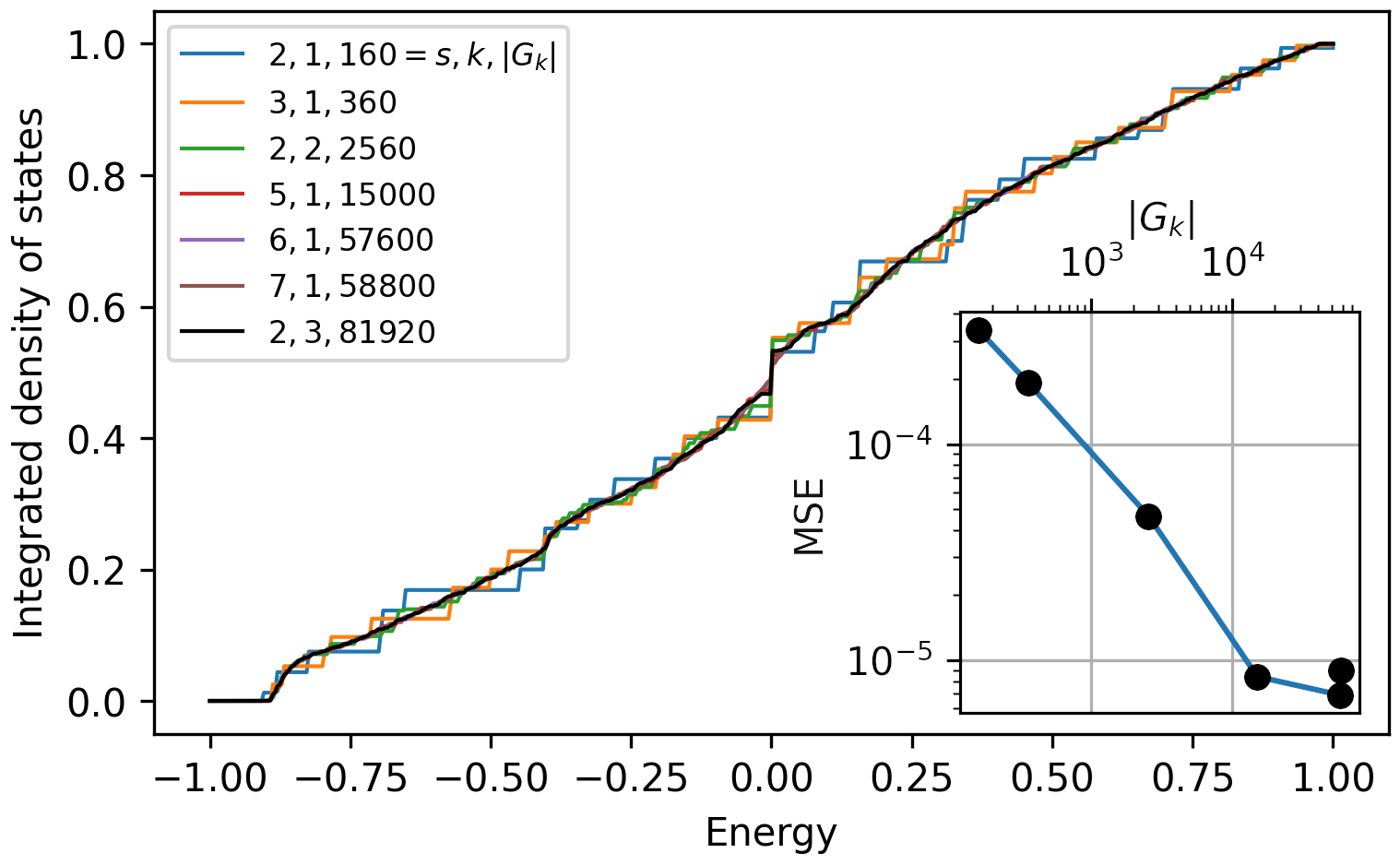}
  \caption{\small Integrated density of states of the adjacency matrix $\Delta$ as function of energy, for various finite-size approximations. The inset shows the mean squared error (MSE) with the largest system as reference.
}\label{Fig:Conv}
\end{figure}

Our tight-binding Hamiltonians will be defined on the point-lattice $\Ll$ from Fig.~\ref{Fig:Tilling}(c), whose points can be uniquely labeled by the elements of $\Delta^+_{\{p,q\}}$. This is natural because the standard left-regular representation of $\Delta^+_{\{p,q\}}$ acts as $\pi_L(g)|g'\rangle = |gg'\rangle$ on the Hilbert space $\ell^2\big (\Delta^+_{\{p,q\}}\big )$ spanned by the vectors $|g\rangle$, $g \in \Delta^+_{\{p,q\}}$, and the Hamiltonians are generated from the right-regular representation of the group algebra $\CM \Delta^+_{\{p,q\}}$:\vspace{-0.2cm}
\begin{equation}\label{Eq:H1}
    H= \pi_R  (h ), \quad h = \sum\nolimits_{g \in \Delta^+_{\{p,q\}}} w_g \cdot g,\vspace{-0.2cm}
\end{equation}
with $w_g=w_{g^{-1}}^\ast \in \CM$. Explicitly, $H$ acts as \vspace{-0.2cm} 
\begin{equation}\label{Eq:RRRep}
H|g'\rangle = \sum\nolimits_{g \in \Delta^+_{\{p,q\}}} w_g |g' g^{-1}\rangle, \vspace{-0.2cm}
\end{equation}
and it is straightforward to check that $\pi_L(g)\pi_R(h) = \pi_R(h) \pi_L(g)$ for any $h$ as in Eq~\eqref{Eq:H1} and $g \in \Delta^+_{\{p,q\}}$. These Hamiltonians can be easily implemented with metamaterials using the coupling-mode theory as guidance \cite{ApigoPRM2018}.

We now explain L\"uck's work \cite{LuekGFA1994,LuekBook2002} on formulating PBCs in algebraic fashion. Any finite-index normal subgroup $N_k$ of $\Delta^+_{\{p,q\}}$ sets a canonical projection onto the finite quotient sub-group $\rho_k : \Delta^+_{\{p,q\}} \to G_k = \Delta^+_{\{p,q\}}/N_k$, which can be lifted to the level of group algebras: \vspace{-0.2cm}
\begin{equation}\label{Eq:Hk}
    h \mapsto \rho_k(h) = \sum\nolimits_{g \in \Delta^+_{\{p,q\}}} w_g \cdot \rho_k(g) \in \CM G_k.\vspace{-0.2cm}
\end{equation}
The right-regular representation $\pi_R$ of $\CM G_k$ acts on the finite Hilbert space $\ell^2(G_k)$ of dimension $|G_k|$ and the spectrum of $\pi_R(h_k)$ is contained in the spectrum of $H$. Therefore, the spectral gaps of $H$ are not contaminated by the so-constructed finite-size approximation \cite{LuxArxiv2023} and this is what the PBCs are sought to deliver! To achieve the thermodynamic limit, one needs a whole coherent sequence of normal subgroups $\Delta^+_{\{p,q\}} =N_0 \triangleright N_1 \triangleright N_2 \cdots$, such that $\cap N_k =\{1\}$ \cite{LuekGFA1994,LuekBook2002}. Then the Green's functions can be recovered with arbitrary precision from the so-constructed finite-size approximations \vspace{-0.2cm}
$$
 \langle g | (H-z)^{-1} |g'\rangle = \lim_{k \to \infty} \langle \rho_k(g) |(\rho_k(H) - z)^{-1}|\rho_k(g') \rangle. \vspace{-0.2cm}
$$
This in turn assures that the spectrum, thermodynamic coefficients, correlation functions, topological invariants, etc., can computed with arbitrary precision from the so-constructed finite-size approximations.

Our observation in \cite{LuxArxiv2023} was that the maps $\rho_k$ can be as straightforward as applying the modulo arithmetic operator on certain coefficients, if $N_k$'s are generated in a specific way. To achieve this for $\Delta^+_{\{p,q\}}$ group, we searched first for a faithful but not necessarily unitary representation of $\Delta_{\{p,q\}}$ with entries landing into the simplest possible ring extension of $\ZM$. Using elements of Coxeter group theory \cite{HumphreysBook1990, Mennicke1967, Siran2001a, Siran2001b}, we found the representation \cite{Supplemental} \vspace{-0.2cm}
 \begin{align*}{\scriptsize
    \sigma_x = \begin{pmatrix}
        -1 & 2 \eta & 0 \\
        0 & 1 & 0 \\
        0 & 0 & 1
    \end{pmatrix},
    \
    \sigma_y = \begin{pmatrix}
        1 & 0 & 0 \\
        2 \eta & -1 & 2\zeta \\
        0 & 0 & 1
    \end{pmatrix},
    \
    \sigma_z = \begin{pmatrix}
        1 & 0 & 0 \\
        0 & 1 & 0 \\
        0 & 2\zeta & -1
    \end{pmatrix},} \vspace{-0.4cm}
\end{align*} 
where $\eta = \cos\big(\frac{\pi}{p}\big ), \zeta = \cos \big ( \frac{\pi}{q} \big )$. Multiplications of these matrices produce entries that all belong to the polynomial ring $\ZM[\xi]$, $\xi = 2\cos\big (\frac{\pi}{pq} \big )$. As such, our representation is a subgroup of ${\rm GL}(3,\ZM[\xi])$. Then the coherent sequence of subgroups can be simply taken as $N_k = \Delta^+_{\{p,q\}} \cap {\rm SL}(3,s^k\, \ZM[\xi])$, $s \in \NM$.

Our last task is to characterize $G_k$'s. As shown in \cite{Supplemental}, $\xi$ is an algebraic number, more precisely it is a root of a minimal irreducible polynomial of degree $\varphi(2p q)/2$, where $\varphi$ is Euler's totient function \cite{Lehmer1933, Watkins1993}. As a consequence, every entry in the matrix representations of $\Delta^+_{\{p,q\}}$ can be written as $\sum_{r=0}^{\varphi/2-1} c_r \xi^r$, with all $c_r$ coefficients from $\ZM$. Multiplication of such series together with the algebraic reduction leads to a specific multiplication of the coefficients $\bm c = \{c_r\}$, which we denote by $\bm c \star \bm c'$ \cite{Supplemental}. Then $G_k$ is the image in the group ${\rm SL}(3,\ZM_{s^k}[\xi])$ of the elements of $\Delta^+_{\{p,q\}}$ under taking the ${\rm mod}\, s^k$ operation on the coefficients $c_r$. Thus, every entry of these matrices can be written as $\sum_{r=0}^{\varphi/2-1} \tilde c_r \xi^r$, with all $\tilde c_r$ coefficients from $\ZM_{s^k}$. As for the multiplication in $G_k$, it is the usual matrix multiplication but with the multiplication between the entries replaced by \vspace{-0.2cm}
$$
\Big (\sum_{r=0}^{\varphi/2-1} \tilde c_r \xi^r\Big )\cdot \Big (\sum_{r=0}^{\varphi/2-1} \tilde c'_r \xi^r\Big ) = \sum_{r=0}^{\varphi/2-1} \big (\tilde{\bm c} \star \tilde{\bm c}')_r \, {\rm mod} \; s^k \  \xi^r.\vspace{-0.2cm}
$$

The order of $G_k$ determines the dimension of the finite Hilbert space $\ell^2(G_k)$, where the approximated Hamiltonians $\rho_k(H)$ from Eq.~\eqref{Eq:Hk} act on. This action goes as in Eq.~\eqref{Eq:RRRep} with $\Delta^+_{\{p,q\}}$ replaced by $G_k$ and, to implement it numerically, all that is needed is an indexing of the elements of $G_k$ and the computation of the multiplication table. A full working code implementing all of that, together with documentation, can be downloaded from \cite{CodeRepository}. It can be seen in action in Fig.~\ref{Fig:Conv}, where the bulk spectrum of the adjacency operator $\Delta = \frac{1}{4}(x_1+x_1^{-1} +x_2 + x_2^{-1})$ is resolved and the exponentially fast convergence to the thermodynamic limit is demonstrated.

\begin{figure}[t]
\center
\includegraphics[width=\linewidth]{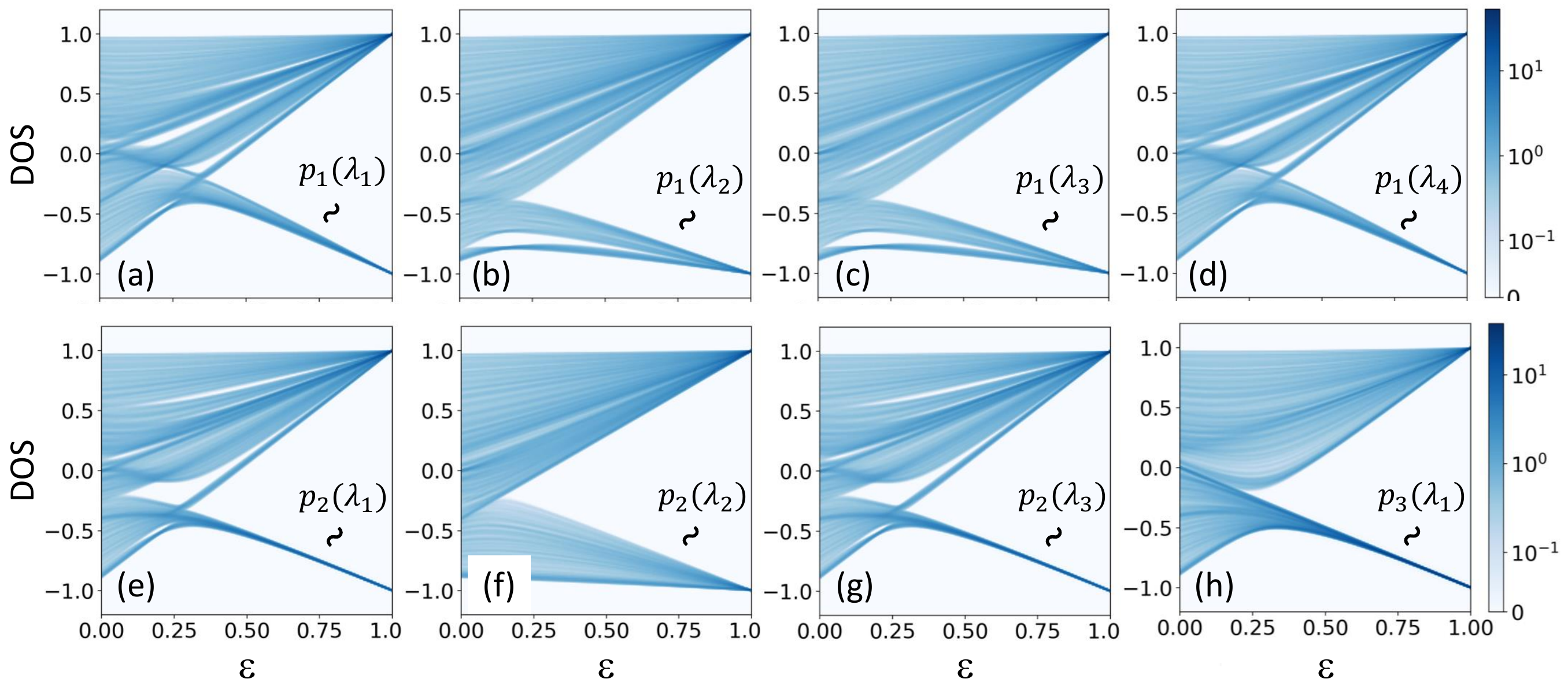}\vspace{-0.2cm}
  \caption{\small Opening available topological gaps using the model Hamiltonians~\eqref{Eq:ModelH}. The topological classes of the lower bands are specified in each panel, where $\lambda_j = e^{\frac{\imath 2 \pi j}{\nu_\alpha}}$. The calculations are performed with PBC generated by $s=2$ and $k=3$. \vspace{-0.6cm}
}
 \label{Fig:TopoGaps}
\end{figure}

We now start our application on resolving the topological bands supported by the $\{p,q\}$ tessellations. The symmetric Hamiltonians live inside the $C^\ast$-algebra of the space-group. For a generic $\Ff_{g,\bm \mu}$ as in Eq.~\eqref{Eq:GenF}, the $K_0$-group of this algebra is isomorphic to $\ZM^{2+\sum_{\alpha=1}^m(\nu_\alpha-1)}$. It is freely generated by the identity operator, a projection $p_0$ that carries the hyperbolic Chern number \cite{CareyCMP1998,MarcolliCCM1999,MarcolliCMP2001} and by the spectral projections of the cyclic elements $x_\alpha$, \vspace{-0.2cm}
\begin{equation}
p_\alpha(\lambda_k) = \frac{1}{\nu_\alpha} \sum_{j=0}^{\nu_\alpha -1}\lambda_k^j x_\alpha^j, \ \lambda_k=e^{\frac{2 \pi \imath k}{\nu_\alpha}}, k=\overline{1,\nu_\alpha}. \vspace{-0.2cm}
\end{equation}
This assures us that any band projection of a symmetric Hamiltonian can be continuously deformed into a stacking of these fundamental projections $p_0^{n_0} \oplus p_1(\lambda_1)^{n_1(\lambda_1)}\oplus \cdots \oplus p_m(\lambda_{\nu_m})^{n_m(\lambda_{\nu_m})}$ without closing the flanking spectral gaps \cite{footnote7}. The integer numbers $\{n_0,n_1(\lambda_1)\ldots,n_m(\lambda_{\nu_m})\}$ represent a complete set of independent topological invariants that can be associated to a band projection.

In the case of $\Delta^+_{\{p,q\}} = \langle 0;p,q,2\rangle$, we have a total of eight $p_\alpha(\lambda)$ projections and, for example, \vspace{-0.2cm}
\begin{equation}\label{Eq:ModelH}
h_\alpha(\lambda,\epsilon)=\epsilon \big (1-2p_\alpha(\lambda)\big) +(1-\epsilon) \Delta \vspace{-0.2cm}
\end{equation}
are model Hamiltonians displaying topological bulk gaps and bottom spectral bands carrying the $K$-theoretic labels $n_\alpha(\lambda)=1$ \cite{footnote9}. Openings of these topological spectral gap are shown in Fig.~\ref{Fig:TopoGaps} and the types of physical couplings needed between the resonators to implement $h_1(\lambda)$ are shown in Fig.~\ref{Fig:Tilling}(d). They can certainly be implemented with the platform developed in \cite{ChenNatComm2023}.

\begin{figure}[t]
\center
\includegraphics[width=0.99\linewidth]{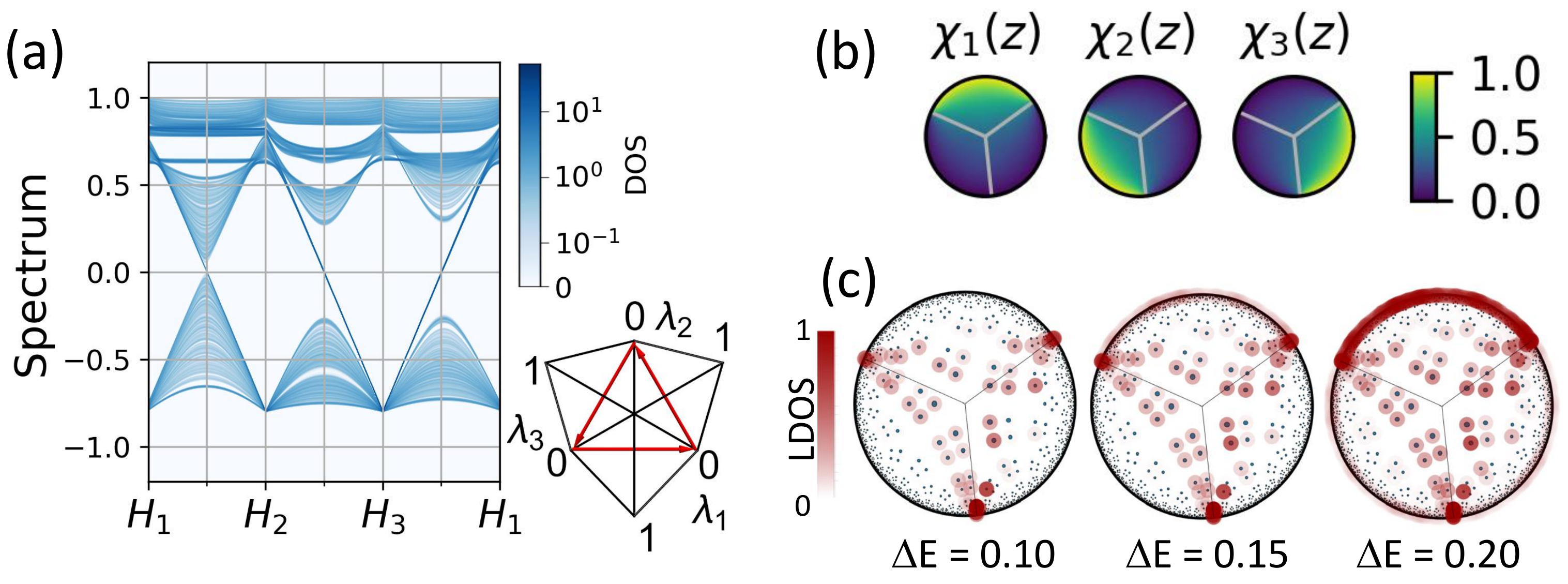}
  \caption{\small (a) Energy spectrum of the Hamiltonian $\lambda_1 H_1+\lambda_2 H_2 + \lambda_3 H_3$ along the path shown in the inset, as computed with the PBC $s=2$ and $k=3$ and with $H_i$ as in Fig.~\ref{Fig:TopoGaps} panels (a), (e), (h), respectively, and $\epsilon=0.8$. (b) Our smooth partition system for the hyperbolic disk. (c) Local density of states (LDOS)~\eqref{Eq:LDOS} of Hamiltonian~\eqref{Eq:YH}, as computed with open BC on a crystal with  14255 sites. 
}
 \label{Fig:Y_junction}
\end{figure}

To demonstrate the distinct topological characters of the models~\eqref{Eq:ModelH}, we illustrate in Fig.~\ref{Fig:Y_junction}(a) samples of energy spectra resulted from pairwise interpolations. The seen topological spectral flows, which actually occur for any pair, demonstrate that two such models cannot be adiabatically connected. Since these topological spectral flows are stable against turning on or off degrees of freedom, they can be used in applications that require robust spectral engineering. Let us point out that witnessing these topological spectral flows would have been impossible without use of PBCs.

Another possible application is engineering soft topological interface channels, which can be achieved by rendering the smooth interpolations from Fig.~\ref{Fig:Y_junction}(a) in the physical space. Due to the large number of available topological phases, we can engineer complex interfaces such as the Y-junction shown in Figs.~\ref{Fig:Y_junction}(b-d). There, we use the smooth partition of the hyperbolic unit disk $\sum_{i=1}^3 \chi_i(z) =1$ shown in Fig.~\ref{Fig:Y_junction}(b), and generate a lattice Hamiltonian $H$ with matrix elements 
\begin{equation}\label{Eq:YH}
    \langle z | H |z'\rangle = \sum\nolimits_{i=1}^3 \chi_i(\mu_{z,z'}) \langle z | H_i |z'\rangle, \quad z,z' \in \Ll, 
\end{equation}
where $\mu_{z,z'}$ is the mid geodesic point of $z$ and $z'$~\cite{Wang2019}. It smoothly interpolates in space between the Hamiltonians $H_1, \, H_2, \ H_3$ showcased in Figs.~\ref{Fig:TopoGaps}(a,e,h), respectively. Figs.~\ref{Fig:Y_junction}(c) shows renderings of the local density of states 
\begin{equation}\label{Eq:LDOS}
    {\rm LDOS}(E,z) = \sum\nolimits_{n} e^{-|E_n-E|^2/(2\Delta E^2)} \; |\psi_n(z)|^2, 
\end{equation}
for several values of $\Delta E$, where the sum is over the eigenstates $H |\psi_n\rangle = E_n |\psi_n\rangle$. The energy $E$ is pinned at $E=0$ in the middle of the common bulk gap of the Hamiltonians $H_i$'s, and where Fig.~\ref{Fig:Y_junction}(a) shows a crossing of topological modes. The plots in Fig.~\ref{Fig:Y_junction}(c) confirm the expected soft topological interface modes between the distinct topological phase.

In conclusion, we derived an algorithmic procedure to impose PBCs on finite hyperbolic crystals of increasing sizes and demonstrated the exponentially fast convergence of the bulk properties to the thermodynamic limit when computed with our algorithms. Our work enables now the identification of the gapped topological phases supported by generic $\{p,q\}$-tessellations of the hyperbolic spaces and simulations of various topological dynamical effects. Work is in progress on how to extend these results in the presence of a magnetic field.

{\it Acknowledgements.} This work was supported by the U.S. National Science Foundation through the grants DMR-1823800 and CMMI-2131760.

\newpage

\section{Supplemental Material}

\subsection{On crystallography of hyperbolic disk} 

Here, we established the connection between the standard exact sequence
\begin{equation}\label{Eq:SES2}
1 \rightarrow \Ff_{g'} \rightarrow \Ff_{g,{\bm \nu}} \rightarrow P=\Ff_{g,{\bm \nu}}/\Ff_{g'} \rightarrow 1, \vspace{-0.2cm}
\end{equation}
and the recent work \cite{BoettcherPRB2022}. First, there is a known relation between $g'$ and the index of the group $P$ appearing in Eq.~\eqref{Eq:SES2}, namely \cite{MathaiATMP2019} 
\begin{equation}\label{Eq:GP}
g'=1+\tfrac{1}{2}|P|\big (2(g-1)+(m-\nu)\big), \quad \nu=\sum_{j=1}^m 1/\nu_j.
\end{equation} 
Now, assume that $\Ff_{g,{\bm \nu}}$ is associated to a $\{p,q\}$-tessellation, as explained in the main text. Then, if $F$ is the number of regular polygons tessellating the fundamental domain of $\Ff_{g'}$ in Eq.~\eqref{Eq:SES2} and $A(p,q)=(p - 2)\pi - 2 \pi p/q$ is the hyperbolic area of the regular polygon, then Gauss-Bonnet theorem gives $FA(p, q) = 4\pi (g'-1)$, which can be put into the form
$$
g' = 1 + \tfrac{1}{2}pF \Big(\frac{1}{2} - \frac{1}{p} -\frac{1}{q}\Big).
$$
By comparing with Eq.~\eqref{Eq:GP}, we can calculate the order of the point group as $|P| = p F$. This is exactly the order of the point group considered in \cite{BoettcherPRB2022}[Sec.~IIID] (with reflection excluded). Therefore, we can conclude that lowest allowed value of $g'$ from Eq.~\eqref{Eq:SES2} coincides with $g$ from Table.~IV from \cite{BoettcherPRB2022}. In our opinion, though, the torsion-free Fuschian group $\Ff_{g'}$ is the rightful generator of the Bravais lattice of the crystal and not $\Delta_{\{p_B,q_B\}}^+$. Of course, we are just talking about a lexicon and it remains to be seen which of these views is more fruitful.

\subsection{Representation of triangle groups}

In the following derivations we will follow closely the book of Humphreys \cite{HumphreysBook1990}.
We start by defining the triangle group
\begin{align}
     \Delta_{\lbrace p,q\rbrace } = \braket{x,y,z | x^2,y^2,z^2, (xy)^p, (yz)^q, (zx)^2 }.
\end{align}
We are after a faithful representation of this group with values in a ring-extension of the integers, a result that is originally due to Mennicke \cite{Mennicke1967}.
In this notation, $\Delta_{\lbrace p,q\rbrace }$ is the quotient of the free group generated by $x$, $y$ and $z$ by the smallest normal subgroup containing the elements $x^2$, $y^2$, $z^2$, $(xy)^p$, $(yz)^q$, and $(zx)^2$. The generators $x$, $y$ and $z$ can be interpreted as reflections across the sides of a triangle with the rational internal angles
\begin{equation}
    \alpha = \pi/p, ~ \beta = \pi/q~ , \gamma = \pi/2.
\end{equation}
The mental model in our case are the right-triangles in the hyperbolic disk which serve as the fundamental domain of hyperbolic tesselations.
We additionally assume that $0< p,q < \infty $ and that the angle inqueality for hyperbolic triangles holds, i.e.,
\begin{equation}
    \alpha + \beta + \gamma < \pi.
\end{equation}
The application of subsequent reflections across the sides of the triangle which are incident on a given vertex will generate a rotation around the given vertex.
The angle of this rotation is twice of the respective internal angle $\alpha$, $\beta$ or $\gamma$.
Since the angles are rational, one can see how this can lead to torsion elements $(xy)^p$, $(yz)^q$ and $(zx)^2$.
The triangle group above can be seen as the full space group symmetry of a potential hyperbolic crystal that we would like to construct.

The triangle group $ \Delta_{\lbrace p,q\rbrace }$ is a special case of a Coxeter group, i.e., a group with presentation
\begin{align}
    \braket{g_1, g_2, \ldots, g_n |  (g_i g_j)^{m_{ij}}  },
\end{align}
with the requirement that $ m_{ii} = 1$ and $m_{ij} \geq 2 $ if $i\neq j$. 
Which means that in our case, we have
\begin{align}
    m = \begin{pmatrix}
        1 & p & 2 \\
        p & 1 & q \\
        2 & q & 1
    \end{pmatrix},
\end{align}
with respect to the generators $x$, $y$, and $z$.
One can summarize the structure of a Coxeter group with the help of its Coxeter graph, where $m$ is interepreted as a weighted adjacency matrix and an edge is drawn between vertices if the weights is at least three.
This means in our case, we have
\begin{equation}
    \Delta_{\lbrace p,q \rbrace} =  \circ \overset{p}{-} \circ \overset{q}{-} \circ .
\end{equation}
The goal is to find a faithful representation of $\Delta_{\lbrace p,q \rbrace}$ which allows us to introduce the concept of generalized periodic boundary conditions.
Define the $\mathbb{R}$ vector space $V$ with its basis vectors labelled by the generators of $\Delta_{\lbrace p,q \rbrace}$, i.e., $S=\lbrace x,y,z \rbrace \subset \Delta_{\lbrace p,q \rbrace}$.
A typical element $v$ of $V$ thus has the form $ \vec{v} = a \vec{e}_x + b \vec{e}_y + c \vec{e}_z$.
Introduce the symmetric bilinear form $B \colon V\times V \to \mathbb{R}$ by its action on the basis vectors
\begin{equation}
    B(\vec{e}_s, \vec{e}_{s'}) \equiv - \cos \left( \frac{\pi}{ m_{ss'}} \right).
\end{equation}
It has the property that $ B(\vec{e}_s, \vec{e}_{s}) = 1$ and  $ B(\vec{e}_s, \vec{e}_{s'}) \leq 0$ for $s\neq s'$.
Define the subspaces
\begin{align}
    H_s & = \lbrace \vec{v} \in V ~|~ B(\vec{v},\vec{e}_s) = 0 \rbrace \cong V \setminus \mathbb{R} \vec{e}_{s},
    \\
    V_{ss'} & = \mathbb{R} \vec{e}_s \oplus  \mathbb{R} \vec{e}_{s'}.
\end{align}
Each point of $H_s$ is invariant under the action of $\sigma_s$, which can be seen directly from the definition of $\sigma_s$.
We can now attempt to construct a group action on $V$. Define $\sigma_s \colon V \to V$ by setting
\begin{equation}
    \sigma_s \vec{\lambda} = \vec{\lambda} - B(\vec{e}_s, \vec{\lambda}) \vec{e}_s .
\end{equation}
In the following we want to sketch the proof of the following proposition
\begin{proposition}
    The assignment $s\mapsto \sigma_s$ extends to a faithful representation $\sigma \colon \Delta_{\lbrace p,q \rbrace} \to \mathrm{GL}(V)$.
\end{proposition}
First of all, the group multiplication of $\Delta_{\lbrace p, q \rbrace}$ corresponds to the composition of operations acting on $V$, i.e., we interpret a word $s_1 s_2 s_3 \cdots s_n $ as the composition $ \sigma_{s_1} \circ  \sigma_{s_2} \circ  \cdots \sigma_{s_n}$.
This means we adhere to the convention that we ``read'' a word from right to left.
It is then sufficient to verify the relations $(\sigma_{s_i} \sigma_{s_j})^{m_{ij}}$ in order to show that we indeed have a homomorphism onto $\mathrm{GL}(V)$.

By observing the action of $\sigma$ on the basis vectors, we can identify the matrix representation.
Take $ \vec{v}= \sum_{i=x,y,z}  v_i \vec{e}_i$. 
The application of $\sigma_s$ yields
\begin{align}
    \sigma_s \vec{v}=& \sum_{j=x,y,z}  v_j  ( \sigma_s \vec{e}_j )
     =  \sum_{i=x,y,z}  \left( \sum_{j=x,y,z}  (\sigma_s)_{ij}v_j \right) \vec{e}_i,
\end{align}
where we can identify
\begin{align}
    (\sigma_s)_{ij} &= \delta_{ij} - 2 B (\vec{e}_i, \vec{e}_j) \delta_{si} 
    \notag \\
    &= \delta_{ij} + 2 \cos(\pi/m_{ij}) \delta_{si}.
\end{align}
Note that this construction underlies the equivalence of categories between the category of group modules and the category of group representations over a given field $\mathbb{K}$:
\begin{equation}
    \mathbf{Rep(G)}_{\mathbb{K}} \rightleftarrows \mathbf{GMod}_{\mathbb{K}} .
\end{equation}
In retrospect, we therefore have indeed constructed a suitable $\mathbb{R}\Delta_{\lbrace p,q\rbrace}$-module before extracting the corresponding representation.
In matrix notation, we find
\begin{align}
    \sigma_x &= \begin{pmatrix}
        -1 & 2 \cos(\alpha) & 2\cos(\gamma) \\
        0 & 1 & 0 \\
        0 & 0 & 1
    \end{pmatrix},
    \\
    \sigma_y &= \begin{pmatrix}
        1 & 0 & 0 \\
        2 \cos(\alpha) & -1 & 2\cos(\beta) \\
        0 & 0 & 1
    \end{pmatrix},
    \\
    \sigma_z &= \begin{pmatrix}
        1 & 0 & 0 \\
        0 & 1 & 0 \\
        2 \cos(\gamma) & 2\cos(\beta) & -1
    \end{pmatrix}.
\end{align}
To get Mennicke's original result \cite{Mennicke1967}, replace $\sigma_s \to - \sigma_s^T$.
By construction, one has $\sigma_s \vec{e}_s = - \vec{e}_s$, which we expect from a reflection.
Note also that $\det \sigma_s = -1$.
Further, one finds
\begin{align}
    \sigma_s ( \sigma_s \vec{\lambda})
    & = 
    \sigma_s (\vec{\lambda} - 2 B(\vec{e}_s, \vec{\lambda}) \vec{e}_s) 
    \notag \\
    & = \vec{\lambda}- 4 B(\vec{e}_s, \vec{\lambda})  \vec{e}_s  +4 B(\vec{e}_s,  B(\vec{e}_s, \vec{\lambda}) \vec{e}_s)
    \notag \\
    & = \vec{\lambda} - 4 B(\vec{e}_s, \vec{\lambda})  \vec{e}_s +  4 B(\vec{e}_s, \alpha_{s})  B(\vec{e}_s,\vec{\lambda})  \vec{e}_s
    \notag \\
    & = \vec{\lambda} ,
\end{align}
and therefore $\sigma_s^2 = 1$, verifying the diagonal Coxeter relations.
Further, one finds that $\sigma_{xy} = \sigma_x \sigma_y$
has eigenvalues $1, e^{\pm 2i \alpha}$ and is therefore of order $p$.
Similarly, $\sigma_{yz}$ has eigenvalues $1, e^{\pm 2 i \beta}$ and is of order $q$, while $\sigma_{xz}$ has eigenvalues $1, e^{\pm 2 i \gamma}$ and therefore has order $2$.
Therefore $\sigma$ respects to relations of $\Delta_{\lbrace p,q\rbrace }$.
For a proof of the faithfulness of the representation $\sigma$, we refer to the book Humphreys (\cite{HumphreysBook1990} Section 5.4 p. 113).

\subsection{Representation of proper triangle groups}

The proper triangle group is defined as
\begin{align}
    \Delta_{\lbrace p,q\rbrace }^+ = \braket{A,B |  A^p, B^q, (AB)^2 }.
\end{align}
We can obtain a representation for the generators of $ \Delta_{\lbrace p,q\rbrace }^+$ from our previous knowledge on the geometric representation of $\Delta_{\lbrace p,q\rbrace }$, where we can make the identification
\begin{equation}
    A=xy, ~~B=yz,
\end{equation}
and thereby realize $\Delta_{\lbrace p,q\rbrace }^+$ as a subgroup of  $\Delta_{\lbrace p,q\rbrace }$.
Denote
\begin{align}
    \eta = \cos \alpha, ~~\zeta = \cos \beta, ~~ \delta = \cos \gamma = 0.
\end{align}
Then we find
\begin{align}
A &= \sigma_x \sigma_y =
\left(
\begin{array}{ccc}
 4 \eta ^2-1 & -2 \eta  & 2 \delta +4 \zeta  \eta  \\
 2 \eta  & -1 & 2 \zeta  \\
 0 & 0 & 1 \\
\end{array}
\right)
\\
B &= \sigma_y \sigma_z =
\left(
\begin{array}{ccc}
 1 & 0 & 0 \\
 4 \delta  \zeta +2 \eta  & 4 \zeta ^2-1 & -2 \zeta  \\
 2 \delta  & 2 \zeta  & -1 \\
\end{array}
\right)
\end{align}

\subsection{Relation to the Poincar\'e disk model}

\begin{figure}[t]
    \centering
    \includegraphics[width=0.8\linewidth]{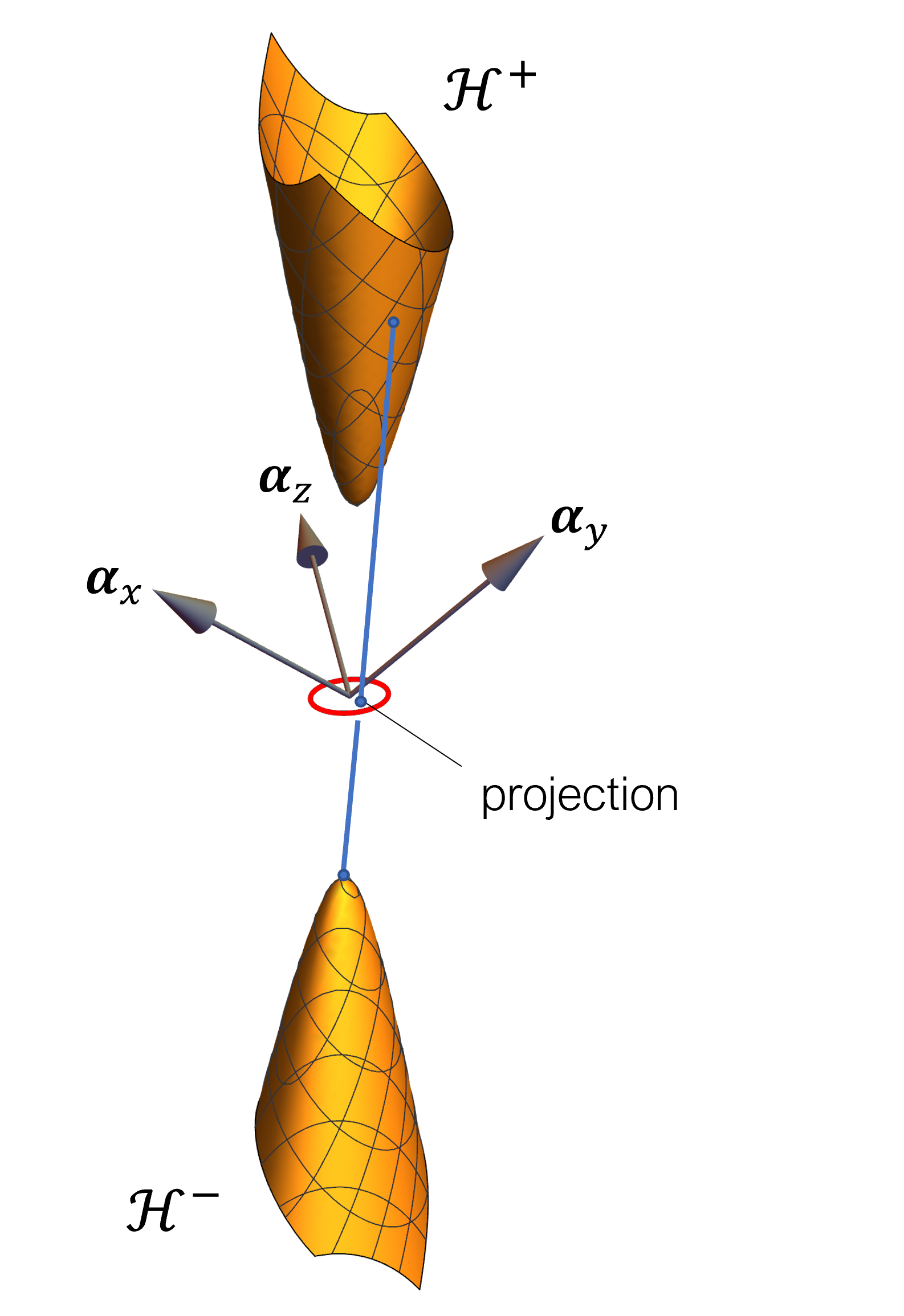}
    \caption{ {\bf The hyperboloid model and the Poincar\'e disk}. Mennicke's bilinear form $B$ can be used to define a two-sheet hyperboloid $\mathcal{H}^+ \cup \mathcal{H}^{-}$ where the positive sheet identifies with $\mathbb{D}$ through the projection which is illustrated in this figure.}
    \label{fig:my_label}
\end{figure}

We can interpret the quadratic form $B$ as a Minkowski metric on $V$.
This is because $B$ has the eigenvectors
\begin{align}
    \boldsymbol{\Gamma}_1 &= \left( -\frac{\zeta}{\eta},0, 1\right)^T, \\ 
    \boldsymbol{\Gamma}_2 &=\left( \frac{\eta}{\zeta},  -\frac{\Upsilon}{\zeta}, 1\right)^T, \\
    \boldsymbol{\Gamma}_3 &= \left( \frac{\eta}{\zeta},  \frac{\Upsilon}{\zeta}, 1\right)^T ,
\end{align}
with respect to the basis defined by $\boldsymbol{\alpha}_s$ and corresponding to the eigenvalues  $1$, $1+\Upsilon $ and $1-\Upsilon$, where $\Upsilon = \sqrt{\eta^2 + \zeta^2}$. 
Note that for $m>2$, one has $ \cos(2\pi/m) > 0$. 
This means that for all $p,q > 2$, $\exists \epsilon > 0$ such that
\begin{equation}
    \frac{1}{2} ( \cos ( 2 \pi / p ) +   \cos ( 2 \pi / q ) ) > \epsilon .
\end{equation}
As a consequence, one finds
\begin{align}
    \Upsilon &= \sqrt{\eta^2 + \zeta^2}
    \notag \\ &=
    \sqrt{\cos^2 (\pi/p) + \cos^2 (\pi/q)}
    \notag \\ &=
    \sqrt{1 + ( \cos (2\pi/p) + \cos (2\pi/q)/2}
    \notag \\ & > 
    \sqrt{1 + \epsilon}
    \notag \\ & > 1 ,
\end{align}
making using of the strict monotonicity of the square root. 
This means, if we parameterize $\vec{v}\in V$ as
\begin{align}
    \vec{v} &= q_1 \boldsymbol{\Gamma}_1  + \frac{q_2}{\sqrt{\Upsilon+1}} \boldsymbol{\Gamma}_2
    + \frac{q_0}{\sqrt{\Upsilon-1}} \boldsymbol{\Gamma}_3
    \notag \\
    &\equiv (q_0,q_1,q_2)_\Gamma .
\end{align}
with $q_0,q_1,q_2\in \mathbb{R}$, the quadratic form evaluates to
\begin{align}
    B(\vec{v},\vec{v}) = q_1^2 + q_2^2 - q^0_2 ,
\end{align}
the Minkowski quadratic form. Define a subset of $V$ by
\begin{align}
    \mathcal{H} &= \lbrace \vec{v} \in V ~|~ B(\vec{v},\vec{v})=-1 \rbrace 
    \notag \\
    & \cong \lbrace 
    (q_0,q_1,q_2)_\Gamma \in V ~|~  q_1^2 + q_2^2 - q^0_2 = -1
    \rbrace 
\end{align}
By construction, this surface is invariant under the action of the triangle group: $\Delta_{\lbrace p,q \rbrace } \mathcal{H} = \mathcal{H}$.
We have therefore shown that

\begin{proposition}
    The subset given by all $\vec{v}\in V$ such that $B(\vec{v},\vec{v})=-1$ defines a disconnected $\Delta_{\lbrace p,q \rbrace }$-invariant surface $\mathcal{H}$ in $V$ which is isomorphic to the two-sheet hyperboloid.
\end{proposition}
Define the positive sheet of this hyperboloid by
\begin{align}
    \mathcal{H}^+ = \lbrace (q_0,q_1,q_2)_\Gamma \in \mathcal{H} ~|~ q_0 > 0 \rbrace .
\end{align}
Then there is a well-known invertible mapping from $\mathcal{H}^+$ to the hyperbolic disk $\mathbb{D} \cup \lbrace \infty \rbrace$, given by the projection
\begin{equation}
   \mathcal{H}^+ \ni (q_0,q_1,q_2)_\Gamma \mapsto  \frac{q_1 + i q_2}{1 + q_0} \in \mathbb{D} \cup \lbrace \infty \rbrace.
\end{equation}

\subsection{Understanding the number field which underlies the geometric representation}

Our goal is to find the minimal field extension of the rationals which is capable of describing the generators of the geometric representation of $\Delta_{\lbrace p,q \rbrace}$.
To do so, we first need to introduce some notation.
Let $T_n \colon [-1,1] \to \mathbb{R}$ represent the Chebyshev polynomial of the first kind, i.e.,
$T_n( \cos \theta ) = \cos n \theta$ for all $\theta \in [0,2\pi]$ and introduce the re-scaled polynomials $P_n(x) = 2 T_n(x/2)$.
Since $T_n$ fulfills the recursion relations $T_0(x) = 1$, $T_1(x)=x$ and $T_{n+1}(x) =  2x T_n(x) - T_{n-1}(x)$, $P_n$ obeys
\begin{align}
    P_0(x) &=2 \\
    P_1(x) &=x \\
    P_{n+1}(x) &=  x P_n(x) - P_{n-1}(x) .
\end{align}
In particular, this implies that $P_n(x)$ (for $n>0$) is a polynomial of degree $n$ with leading term $1 x^n$. All coefficients from $P_n(x)$ are integer such that $P_n(x) \in \ZM[x]$.
Since the product of Chebyshev polynomials is given by $ T_n(x)  T_m(x) = ( T_{n+m}(x) - T_{|n-m|}(x) ) /2$, the corresponding rule
for $P_n$ reads
\begin{equation}
    P_n(x) P_m(x) = P_{n+m}(x)-P_{|n-m|}(x) .
\end{equation}
We can now make the following observation:
\begin{proposition}
    $\xi_n = 2 \cos( 2\pi / n )$ is an algebraic number.
\end{proposition} 
\begin{proof}
Using the properties of the Chebyshev polynomial, one has
\begin{equation}
    T_{n}(\xi_n/2) = \cos( n 2\pi / n ) = 1.
\end{equation}
Rearranging this result, we find that $h(\xi) = 2T_{n}(\xi_n/2) -2 =  P_n(\xi_n) -2 =0 $.
Therefore $\xi$ is a root of the polynomial $h(x) \in \ZM[x]$ and is therefore algebraic.
\end{proof}
It is now tempting to assume that $\ZM[\xi_n]$ is a free module with basis $\xi_n^0, \xi_n^1, \ldots, \xi_n^{n-1}$.
However, there are in general more relations among the $\xi_n$ since the polynomial $h(x)$ is in fact not irreducible over the integers.
For example, if $n=2s$, we can directly find another relation. 
This is because $ T_{s}(\xi_{2s}/2) = \cos( s \pi / s ) = -1$ or equivalently $P_s(\xi) + 2  = 0$.
By the multiplication rule, and using $P_0(x)=2$, one finds
\begin{align}
    P_s(x) P_s(x)  = P_{2s}(x)  - 2 = h(x),
\end{align}
which shows quite explicitly that it can be possible to split up the polynomial further over the integers.

A full reduction of the polynomial can be accomplished which follows from the properties of cyclotomic fields which we will very briefly introduce.
Denote the $n$-th root of unity by $\zeta_n = e^{2\pi i /n }$.
Then $\xi_n = \zeta_{n} + \zeta_{n}^{-1} $.
The field extension $\QM(\zeta_n) $ is known as the cyclotomic field.
The algebraic number field $\QM(\zeta_{n} + \zeta_{n}^{-1})$ is the maximal real subfield of $\QM(\zeta_n) $.
It contains $\ZM[ \zeta_{n} + \zeta_{n}^{-1} ]$ as its ring of integers and it is this ring that we would like to understand better.

By definition, $\zeta_n$ is a root of the polynomial $x^n-1$.
If $n$ is divisible by $d$, then $\zeta_d$ is another root of $x^n-1$ but also of $x^d-1$.
An $n$-th root of unity is primitive if it is not also an $m$-th root of unity for $m<n$.
The primitive $n$-th roots of unity are therefore given by
\begin{equation}
   \mathcal{R}_n=\lbrace \zeta_n^k ~|~1 \leq k \leq n, \mathrm{gcd}(k,n)=1 \rbrace .
\end{equation}
One has $|\mathcal{R}_n| = \varphi(n)$, where $\varphi(n)$ is Euler's totient function, counting the positive integers smaller or equal to $n$ which are coprime to $n$.
Define the polynomial
\begin{align}
    \Phi_n(x) = \prod_{\zeta \in  \mathcal{R}_n} ( x-\zeta ).
\end{align}
This polynomial is also known as the cyclotomic polynomial.
The roots of $\Phi_n(x)$ are given by the primitive $n$-th roots of unity and $\Phi_n(x)$ has degree $\varphi(n)$.
$\Phi_n(x)$ is known to be irreducible in both $\mathbb{Q}[x]$ and $\mathbb{Z}[x]$.

A polynomial $f(x) \in \ZM[x]$ is called palindromic if $f(x)^\ast \equiv x^{n} f(x^{-1}) = f(x)$, where $n$ is the degree of $f$ and it turns out that
\begin{proposition}
    $\Phi_n(x)$ is palindromic for $n>2$.
\end{proposition}
\begin{proof}
    We check that the condition for palindromic polynomials is fulfilled:
    \begin{align}
        x^{\varphi(n)} \Phi_n(x^{-1})
        & = \prod_{\zeta \in  \mathcal{R}_n} ( 1-\zeta x )
        \notag \\
        & = \prod_{\zeta \in  \mathcal{R}_n} \zeta \prod_{\zeta \in  \mathcal{R}_n} ( \zeta^{-1}- x ) .
    \end{align}
    For $n>2$, the product $\prod_{\zeta \in  \mathcal{R}_n} \zeta $ evaluates to $1$. 
    This is because for $\xi_n^k \in \mathcal{R}_n$ one also has $\xi^{n-k}_n \in \mathcal{R}_n $ which cancel each other mutually in the product.
    Since $\varphi(n) \mod 2 = 0$ for $n>2$, one can then write
    \begin{align}
        x^{\varphi(n)} \Phi_n(x^{-1})
        & = (-1)^{\varphi(n)}\prod_{\zeta \in  \mathcal{R}_n} ( x-\zeta^{-1}) 
        \notag \\
        & = \prod_{\zeta \in  \mathcal{R}_n} ( x-\zeta) = \Phi_n(x),
    \end{align}
    and $\Phi_n(x)$ is palindromic.
\end{proof}

% and due to a theorem of Lehmer \cite{X} it has 
% \begin{equation}
%     \xi_n^0,  \xi_n^1, \xi_n^2, \cdots, \xi_n^{\varphi(n)/2-1},
% \end{equation}
% as its integral basis, where $\varphi(n)$ is Euler's totient function and where $n>3$.
% This means there exists a minimal, irreducible polynomial with integer coefficients, which is of degree $\phi(n)/2$, and which has $\xi_n$ as a root.
% Let this polynomial be denoted by $\Psi_n(x)$.
% To understand the foundation for Lehmer's theorem, one can make the following observation

\begin{proposition} \label{proposition:palindroms}
    Let $f(x) \in \ZM[x]$ be palindromic, i.e, $f(x)^\ast \equiv x^{2n} f(x^{-1}) = f(x)$, where $2n$ is the degree of $f$.
    Then there exists a $g(x) \in \ZM[x]$ such that $x^{-n}f(x) = g(x + x^{-1})$ and where $g$ is of degree $n$.
\end{proposition}

\begin{proof}
    Write $f(x) = \sum_{k=0}^{2n} a_k x^k $ with $a_k \in \ZM$.
    Since $f$ is palindromic, one has $a_{2n-k}=a_k$. 
    Then we give an explicit construction for $g$, by
    \begin{align}
        x^{-n} f(x)
        & = \sum_{k=0}^{2n} a_k x^{k-n} 
        \notag \\
       %  & = a_n + \left( \sum_{k=0}^{n-1} a_k x^{k-n} \right) 
       %  +
       %   \left( \sum_{k=n+1}^{2n} a_k x^{k-n} \right) 
       %  \notag \\
       % & = a_n +  \sum_{k=0}^{n-1} a_k x^{k-n} + a_{k+n+1} x^{k+1}
       %  \notag \\
       % & = a_n +  \sum_{k=1}^{n} a_{k-1} x^{k-n-1} + a_{k+n} x^{k}  
       % \notag \\
       % & = a_n +  \sum_{k=1}^{n} a_{n-k} x^{-k} + a_{k+n} x^{k} 
       % \notag \\
       & = a_n +  \sum_{k=1}^{n}  a_{k+n} (x^{k} + x^{-k})
       \notag \\
       & = a_n +  \sum_{k=1}^{n}  a_{k+n} P_k(x + x^{-1}).
    \end{align}
    The last step follows from the identity $P_m(x+x^{-1}) = x^m + x^{-m}$.
    Therefore, the right-hand side is a polynomial in $x+x^{-1}$ with integer coefficients and it is of degree $n$.
\end{proof}

We are now ready to understand the elements of the proof of the following theorem

\begin{theorem}\label{theorem:lehmer}[Lehmer \cite{Lehmer1933}]
    There exists a minimal irreducible polynomial $\Psi_n(x)\in \ZM[x]$ which has $\xi_n$ as its root and which is of degree $\varphi(n)/2$ for $n>2$.
\end{theorem}
\begin{proof}
    We outline the proof given by Lehmer. By the previous proposition, we know that there exists a $\Psi_n(x)\in \ZM[x]$ such that for $n>2$
    \begin{align}
        x^{-\varphi(n)/2} \Phi_n(x) = \Psi_n(x + x^{-1} ) ,
    \end{align}
    since $\Phi_n$ is palindromic for $n>2$.
    Therefore $\Psi_n$ has degree $\varphi(n)/2$.
    The irreducibility of $\Psi_n$ follows from the irreducibility of $\Phi_n$.
    Since $\zeta_n^k \in \mathcal{R}_n$  is a root of $\Phi_n$, $\zeta_n^k + \zeta_n^{-k}$ is a root of $\Psi_n$.
    Therefore  $\Psi_n(x)$ is the minimal polynomial that has $\xi_n$ as its root.
\end{proof}

\begin{corollary}
    $\ZM[\xi_n] $ is a free $\ZM$-module with basis $\xi_n^0, \xi_n^1, \cdots \xi_n^{\varphi(n)/2-1}$ . 
    Therefore $\ZM[\xi_n] \cong_\ZM \ZM^{\varphi(n)/2}$.
\end{corollary}

\begin{theorem} [Watkins and Zeitlin \cite{Watkins1993}] 
    The minimal polynomial $\Psi_n(x) \in \ZM[x]$ which has $\xi_n = 2 \cos(2\pi/n)$ as a root is implicitly given by
    \begin{align}
     \prod_{d|2s} \Psi_d(x ) &=    P_{s+ 1}(x) - P_{s- 1}(x),  %= P_s(x) P_1(x) 
     \\ 
     \prod_{d|2s+1} \Psi_d(x ) &=  P_{s+ 1}(x) - P_{s}(x) .
    \end{align}
\end{theorem}

This allows for a recursive algorithm. 
For example, in the WolframLanguage, the following code could be used
\begin{verbatim}    
P[n_] := 2 ChebyshevT[n, x/2];

Psi[n_]:=Psi[n]=(
  div = Drop[Divisors[n],-1];
  rhs = If[
    EvenQ[n], 
    P[n/2+1]-P[n/2-1],  
    P[(n-1)/2+1]-P[(n-1)/2]
  ];
  
  denominator = Times@@(Psi[#]&/@div);
	
  Return[
    Expand[
      FullSimplify[ 
        rhs/denominator 
      ]
    ]
  ]
);
\end{verbatim}
Using this method we find that

\begin{proposition}
    The polynomial
    \begin{equation}
       \Psi_{40}(x) =x^8-8 x^6+19 x^4-12 x^2+1,
    \end{equation}
    is the minimal polynomial which has $\xi_{2pq}$ as its root for $p=5$ and $q=4$.
\end{proposition}
\begin{proof}
    We have $2pq = 40$. The divisors of $40$ are  $1, 2, 4, 5, 8, 10, 20$ and  $40$.
    The divisor of $20$ are $1, 2, 4, 5, 10$ and $20$.
    Therefore,
    \begin{align}
        \frac{\prod_{d|40} \Psi_d(x ) }{\prod_{d|20} \Psi_d(x ) } = \Psi_8(x ) \Psi_{40}(x ) .
    \end{align}
    The left hand side can be computed from the theorem. 
    We just need to know $\Psi_8(x )$.
    Applying the same trick:
    \begin{align}
        \Psi_8(x ) &=  \frac{\prod_{d|8} \Psi_d(x ) }{\prod_{d|4} \Psi_d(x ) } 
         = \frac{  P_{5}(x) - P_{3}(x) }{  P_{3}(x) - P_{1}(x) }
         = x^2 - 2 .
    \end{align}
    This means we have
    \begin{align}
       \Psi_{40}(x )
        & = 
         \frac{1}{ ( x^2 - 2) } \frac{  P_{21}(x) - P_{19}(x) }{  P_{11}(x) - P_{9}(x) } 
        \notag \\
        & =   \left(x^8-8 x^6+19 x^4-12 x^2+1\right).
    \end{align}
\end{proof}

\begin{proposition}
    $\Psi_n(x)$ is monic for all $n \in \NM$, i.e., the coefficient of its leading term is equal to one.
\end{proposition}
\begin{proof}
    By constructions, the cyclotomic polynomials $\Phi_n(x)$ are monic for all $n \in \NM$.
    Since the re-scaled polynomials $P_n$ are monic as well, it follows from proposition \ref{proposition:palindroms} and theorem \ref{theorem:lehmer} that all $\Psi_n(x)$ are monic as well for $n>2$.
    For $n=1$, one has $\xi_n = 2$ and $\Psi_1(x)=x-2$ is monic as well. This is also true for the case $n=2$, where $\xi_n = -2$ and $\Psi_2(x)=x+2$ is monic.
\end{proof}

Since  $\Psi_n(x)$ is monic and $\Psi_n(\xi_n)=0$ with $\mathrm{deg}~ \Psi_n(x) = \varphi(n)/2$, we can easily express $\xi_n^{\varphi(n)/2}$ in terms of a polynomial with smaller degree.
This is because for $n>2$, we can define a polynomial $r_n(x) \in \ZM[x]$ by
\begin{align}
    r_n(x) = x^{\varphi(n)/2} - \Psi_n(x),
\end{align}
where we are guaranteed that $\mathrm{deg}~r_n(x) \leq \varphi(n)/2-1$ and where 
$r_n(\xi_n)$ is the reduced representation of $\xi_n^{\varphi(n)/2}$ in $\ZM[\xi_n]$ since $r_n(\xi_n)=\xi_n^{\varphi(n)/2}$.
For $n\leq 2$, no reduction is necessary since $\xi_n = \pm 2$.
Let $d= \varphi(n)/2$ and $d-2-k\geq0$, then
\begin{equation}
    \xi_n^{2d-2-k} = \xi_n^d \xi^{d-2-k} = \sum_{l=0}^{d-1} r_l \xi_n^{d-2-k+l},
\end{equation}
where we have introduced $r_n(\xi_n) = \sum_{l=0}^{d-1} r_l \xi_n^l$.
The degree of $\xi_n^{2d-2-k}$ (when viewed as a polynomial in $\xi_n$) has thus been reduced.
This can be used to implement the multiplication law of the ring extension on a computer by folding back each $\xi_n^{2d-2-k}$ starting from the largest exponent, followed by a gradual decent until all terms beyond the $\xi^{d-1}$ have been reduced.
As it has been eluded to earlier, we have a set-theoretic isomorphism 
$\psi\colon \mathbb{Z}^{\varphi(n)/2} \to \mathbb{Z}[\xi_n]$, which assigns a polynomial of degree $\varphi(n)/2$ to a set of $\varphi(n)/2$ integer coefficients. We can use $\psi$, in order to introduce a ring-structure  $\mathbb{Z}^{\varphi(n)/2}$ and therefore obtain a ring-isomorphism $(\mathbb{Z}[\xi_n], \cdot, +) \cong (\mathbb{Z}^{\varphi(n)/2}, \star, +)$,  where $\cdot$ and $+$ denote the product and addition of polynomials in $\mathbb{Z}[\xi_n]$.
While the addition on $\mathbb{Z}^{\varphi(n)/2}$ is given by the natural element-wise addition of the individual integer coefficients (which we denote by the same symbol for convenience), the induced $\star$-product on $\ZM^{\varphi(n)/2}$ is nontrivial due to the intricate back-folding. 
We want the $\star$-product to be consistent with the ring-structure of $\mathbb{Z}[\xi_n]$ and therefore require that
\begin{equation}
    \label{eq:ring_multiplication}
% https://q.uiver.app/?q=WzAsNCxbMCwxLCJcXG1hdGhiYntafVtcXHhpX25dIFxcdGltZXMgXFxtYXRoYmJ7Wn1bXFx4aV9uXSJdLFswLDAsIlxcbWF0aGJie1p9XntcXHZhcnBoaShuKS8yfSBcXHRpbWVzIFxcbWF0aGJie1p9XntcXHZhcnBoaShuKS8yfSJdLFsxLDEsIlxcbWF0aGJie1p9W1xceGlfbl0iXSxbMSwwLCJcXG1hdGhiYntafV57XFx2YXJwaGkobikvMn0iXSxbMSwwLCJcXHBzaSBcXHRpbWVzIFxccHNpIiwyXSxbMCwyLCJcXGNkb3QiLDJdLFsxLDMsIlxcc3RhciJdLFsyLDMsIlxccHNpXnstMX0iLDJdXQ==
\begin{tikzcd}
	{\mathbb{Z}^{\varphi(n)/2} \times \mathbb{Z}^{\varphi(n)/2}} & {\mathbb{Z}^{\varphi(n)/2}} \\
	{\mathbb{Z}[\xi_n] \times \mathbb{Z}[\xi_n]} & {\mathbb{Z}[\xi_n]}
	\arrow["{\psi \times \psi}"', from=1-1, to=2-1]
	\arrow["\cdot"', from=2-1, to=2-2]
	\arrow["\star", from=1-1, to=1-2]
	\arrow["{\psi^{-1}}"', from=2-2, to=1-2]
\end{tikzcd}
\end{equation}
This means that for $\vec{c}, \vec{c}' \in \ZM^{\varphi(n)/2}$ we define
\begin{align}
    \vec{c} \star \vec{c}' = \psi^{-1}( \psi( \vec{c} ) \cdot  \psi( \vec{c}' ) ). 
\end{align}

To conclude the discussion, we can define the general linear group as
\begin{align}
    \mathrm{GL}(3, \ZM[\xi_{2pq}]) = \lbrace g \in \mathcal{M}_{3}( \ZM[\xi_{2pq} ] )  ~|~ g ~\text{invertible}  \rbrace.
\end{align}
and the special linear group by
\begin{align}
    \mathrm{SL}(3, \ZM[\xi_{2pq}]) = \lbrace g \in \mathrm{GL}(3, \ZM[\xi_{2pq}]) ~|~ \det  g = 1  \rbrace.
\end{align}
Note that invertibility in $\mathcal{M}_{3}( \ZM[\xi_{2pq} ] )$ is nontrivial since $\ZM[\xi_{2pq} ]$ is not a field: for $g \in \mathcal{M}_{3}( \ZM[\xi_{2pq} ] )$ to be invertible, one needs $\det  g$ to be invertible in $\ZM[\xi_{2pq} ]$.
Given these definitions,  we reach the following conclusion:
\begin{proposition}
    The image of the geometric representation $\sigma$ lands in $ \mathrm{GL}(3, \ZM [\xi_{2pq}] )$.
    We therefore realize $\Delta_{\lbrace p,q\rbrace}$ as a subgroup of $\mathrm{GL}(3, \ZM [\xi_{2pq}] )$ and
    $\Delta_{\lbrace p,q\rbrace}^{+}$ as a subgroup of $\mathrm{SL}(3, \ZM [\xi_{2pq}] )$.
\end{proposition}

\begin{proof}
    Due to the identities
    \begin{align}
    P_p(\xi_{2pq}) &= 2 T_{p}(\xi_{2pq}/2 )   = 2\cos( \beta) ,
    \\
    P_q(\xi_{2pq}) &= 2 T_{q}(\xi_{2pq}/2 )  = 2\cos( \alpha) ,
\end{align}
    the matrix elements of the generators in the  geometric representation take values in the ring $\ZM[\xi_{2pq}] \subset \QM (\xi_{2pq})$.
\end{proof}
Using these two identities, we can re-interpret the previously obtained matrix matrix representation in the following way:
\begin{align}
    \sigma_x &= \begin{pmatrix}
        -1 & P_q(\xi_{2pq}) & 0 \\
        0 & 1 & 0 \\
        0 & 0 & 1
    \end{pmatrix},
    \\
    \sigma_y &= \begin{pmatrix}
        1 & 0 & 0 \\
        P_q(\xi_{2pq}) & -1 & P_p(\xi_{2pq}) \\
        0 & 0 & 1
    \end{pmatrix},
    \\
    \sigma_z &= \begin{pmatrix}
        1 & 0 & 0 \\
        0 & 1 & 0 \\
        0 & P_p(\xi_{2pq}) & -1
    \end{pmatrix}.
\end{align}
And we again set $A = \sigma_x \sigma_y  $ and $B=\sigma_y \sigma_z$.
Note that for case $q=3$ (or $p=3$), one has $2 \cos(\pi/3) = 1$.
For example, for $q=3$, it would then be sufficient to consider the ring extension $\ZM[\xi_{2p}]$.
Since $\xi_{2p} = 2 \cos(\alpha)$, one then simply has
\begin{align}
    \sigma_x &= \begin{pmatrix}
        -1 & \xi_{2p} & 0 \\
        0 & 1 & 0 \\
        0 & 0 & 1
    \end{pmatrix},
    \\
    \sigma_y &= \begin{pmatrix}
        1 & 0 & 0 \\
        \xi_{2p} & -1 & 1 \\
        0 & 0 & 1
    \end{pmatrix},
    \\
    \sigma_z &= \begin{pmatrix}
        1 & 0 & 0 \\
        0 & 1 & 0 \\
        0 & 1& -1
    \end{pmatrix}.
\end{align}
Another simplification occurs, when $p=q$. 
In this case, it is again sufficient to consider the
 extension $\ZM[\xi_{2p}]$ with
\begin{align}
    \sigma_x &= \begin{pmatrix}
        -1 & \xi_{2p} & 0 \\
        0 & 1 & 0 \\
        0 & 0 & 1
    \end{pmatrix},
    \\
    \sigma_y &= \begin{pmatrix}
        1 & 0 & 0 \\
        \xi_{2p} & -1 & \xi_{2p} \\
        0 & 0 & 1
    \end{pmatrix},
    \\
    \sigma_z &= \begin{pmatrix}
        1 & 0 & 0 \\
        0 & 1 & 0 \\
        0 & \xi_{2p}& -1
    \end{pmatrix}.
\end{align}

\begin{table*}[t]
\caption{ {\bf Irreducible polynomials for $\xi_{2pq}=2 \cos(\pi/(pq))$}. For $p\leq 8$ and $q \leq p$, we list all admissible hyperbolic triangles, sorted by ascending hyperbolic area.
The last column contains the minimal irreducible polynomial $\Psi_{2pq}(x) $, which has $\xi_{2pq}$ as a root.}
\centering
\scalebox{1}{
\begin{tabular}{ccccl}
\toprule
$p$  & $q$ & Area & $|\Delta^+_{\lbrace p,q\rbrace }\mod 2|$ & $\Psi_{2pq}(x) $ for $p\neq q \neq 3$  and $\Psi_{2p}(x) $  for $q=3$ or $p=q$\\
\midrule
$7$ & $3$ & $\frac{\pi }{42}$ & 504 & $x^3-x^2-2 x+1$ \\ 
$8$ & $3$ & $\frac{\pi }{24}$ & 96 & $x^4-4 x^2+2$ \\ 
$5$ & $4$ & $\frac{\pi }{20}$ & 160 & $x^8-8 x^6+19 x^4-12 x^2+1$ \\ 
$6$ & $4$ & $\frac{\pi }{12}$ & 24 & $x^8-8 x^6+20 x^4-16 x^2+1$ \\ 
$5$ & $5$ & $\frac{\pi }{10}$ & 80 & $x^2-x-1$ \\ 
$7$ & $4$ & $\frac{3 \pi }{28}$ & 896 & $x^{12}-12 x^{10}+53 x^8-104 x^6+86 x^4-24 x^2+1$ \\ 
$8$ & $4$ & $\frac{\pi }{8}$ & 16& $x^{16}-16 x^{14}+104 x^{12}-352 x^{10}+660 x^8-672 x^6+336 x^4-64 x^2+2$ \\ 
$6$ & $5$ & $\frac{2 \pi }{15}$ & 60& $x^8-7 x^6+14 x^4-8 x^2+1$ \\ 
$7$ & $5$ & $\frac{11 \pi }{70}$ & 262080 & $x^{12}+x^{11}-12 x^{10}-11 x^9+54 x^8+43 x^7-113 x^6-71 x^5+110 x^4+46 x^3-40 x^2-8 x+1$ \\ 
$6$ & $6$ & $\frac{\pi }{6}$ & 12 & $x^2-3$ \\ 
$8$ & $5$ & $\frac{7 \pi }{40}$ & 2560 & $x^{16}-16 x^{14}+104 x^{12}-352 x^{10}+659 x^8-664 x^6+316 x^4-48 x^2+1$ \\ 
$7$ & $6$ & $\frac{4 \pi }{21}$ & 504 & $x^{12}-11 x^{10}+44 x^8-78 x^6+60 x^4-16 x^2+1$ \\ 
$8$ & $6$ & $\frac{5 \pi }{24}$ & 96 & $x^{16}-16 x^{14}+104 x^{12}-352 x^{10}+660 x^8-672 x^6+336 x^4-64 x^2+1$ \\ 
$7$ & $7$ & $\frac{3 \pi }{14}$ & 448 & $x^3-x^2-2 x+1$ \\ 
$8$ & $8$ & $\frac{\pi }{4}$ & 32 & $x^4-4 x^2+2$ \\ 
\bottomrule
\end{tabular}}
\label{table:irreducible_polynomials}
\end{table*}

\subsection{Definition of the Y-junction}

In this section, we explain the details on how the Y-junction in the manuscript is constructed. 
The goal is to separate the hyperbolic disk into three distinct regions with different topological phases, by assigning the three Hamiltonians
\begin{align}
    H_1 & = \epsilon (1-2 P_A) + (1-\epsilon)  \Delta \\
    H_2 & = \epsilon (1-2 P_B) + (1-\epsilon)  \Delta \\
    H_3 & = \epsilon (1-2 P_C) + (1-\epsilon)  \Delta  ,
\end{align}
where $\epsilon=0.8$ and where $\Delta = (A+B)/4+ \mathrm{h.c.}$, is the Laplace operator. 
As a first step, we define three regions $Y_i$ bounded by three geodesics $\gamma_i$ which originate at $z=0$  and asymptotically meet the following points at infinity:
\begin{align}
    Y_1^\infty &= e^{i \phi_Y}
    \\
    Y_2^\infty & = e^{2\pi i /3}  Y_1^\infty
    \\
    Y_3^\infty & = e^{2\pi i /3}  Y_2^\infty,
\end{align}
where $\phi_Y$ determines the alignment of the junction. 
We consider only the case $\phi_Y = \alpha /2$ in the manuscript.
Let $\arg z$ be the argument of the complex number $z$.
We take $\arg$ to be the principal branch $-\pi < \arg z \leq \pi $.
We can use $\arg z$, to assign the desired local phase of our Y-junction. 
First, we define the relative angles
\begin{align}
  \kappa_i(z) = \arg z^\ast Y_i^\infty ,
\end{align}
and then we make the assignment
\begin{align}
    p(z) = \left\lbrace 
\begin{array}{ll}
     1,&  \text{for}~\kappa_1(z) \leq 0 ~\text{and}~ \kappa_2(z) > 0\\
     2,&  \text{for}~\kappa_2(z) \leq 0 ~\text{and}~ \kappa_3(z) > 0 \\
     3,&  \text{for}~\kappa_3(z) \leq 0 ~\text{and}~ \kappa_1(z) > 0 
\end{array}
    \right.
\end{align}
% \begin{align}
%   \mathrm{H_1} \colon &  \varphi_1 \leq 0 ~\text{and}~ \varphi_2 > 0\\
%   \mathrm{H_2} \colon & \varphi_2 \leq 0 ~\text{and}~ \varphi_3 > 0\\
%   \mathrm{H_3} \colon &   \varphi_3 \leq 0 ~\text{and}~ \varphi_1 > 0
% \end{align}
The discontinuous function $p(z)$ therefore assigns a label to each of three sectors which we want to equip with the respective Hamiltonians.
In other words, if the $z$-coordinate is deep into region $p(z)=i$, we want the local Hamiltonian to resemble our previously defined $H_i$.
However, instead of an abrupt junction of the different phases, we rather want to consider a smooth transition from one phase to another with a controllable parameter $\ell$ for the width of the domain wall region. 
To accomplish this, we define the Hamiltonian of the resulting Y-junction by
\begin{align}
    H_\mathrm{Y}(z,z') =  \sum_{l=1}^3 \chi_l(\mu(z,z')) H_i(z,z'),
\end{align}
where $\chi_l$ is a smooth partition of unity:
\begin{align}
    \sum_{l=1}^3 \chi_l(z) = 1. 
\end{align}
Note that we now passed to a representation of the Hamiltonian which refers to coordinates in the hyperbolic disk. 
This is possible, by assigning to each group element $g$ of $\Delta^+_{\lbrace p,q\rbrace}$ (or one of its finite subgroups) a position $g z_0$, where $g$ acts via Moebius transformation on the seed coordinate $z_0$.
The function $\mu(z,z')$ assigns the geodesic midpoint of $z$ and $z$', were we refer to the following theorem:
\begin{theorem}[Wang et al. \cite{Wang2019}] Let $z,z' \in \mathbb{D}$, then the hyperbolic midpoint $\mu(z,z') \in \mathbb{D}$ along the geodesic from $z$ to $z$' with the property $d(z,\mu) = d(z',\mu) = d(z,z')/2$ is given by
\begin{equation}
    \mu(z,z') = \frac{z' (1- |z|^2) + z (1- |z'|^2) }{1- |z|^2 |z'|^2 + A[z,z'] \sqrt{(1- |z|^2)(1- |z'|^2)} },
\end{equation}
where $A[z,z']$ is the Ahlfors bracket
\begin{equation}
    A[z,z'] = (1- |z|^2) (1- |z'|^2) + |z-z'|^2 .
\end{equation}
\end{theorem}

A natural way to construct the partition $\chi_l$ explicitly can be obtained by constructing the shortest distance into and out-of a given phase region.
Whenever $-\pi/2 < \kappa_i(z) < \pi/2$ we can construct a right triangle in the hyperbolic plane to calculate the shortest geodesic distance to the phase boundary $Y_i$, which we denote by $\delta_i(z)$.
If $\kappa_i(z)$ is outside of this region, the shortest distance $\delta_i(z)$ to the phase boundary is a straight line from $z$ to 0.
In combination with the law of hyperbolic sines it then follows that
\begin{align}
    \sinh \delta_i(z) = \left\lbrace\begin{array}{ll}
    | \sin\varphi_i | \sinh d(0,|z|) & ,|\kappa_i(z)| < \frac{\pi}{2} \\
     \sinh d(0,|z|)  & ,  \text{otherwise} \\
    \end{array}\right.
\end{align}
Using these tools, we can define $D_i(z)$  as the shortest, signed, distance out-of, or into phase region $i$ depending on the desired region.
Recall that the label of desired local phase is given by $p(z)\in \lbrace 1,2,3 \rbrace$.
%Denote by $n_1(z)$ and $n_2(z)$ the two remaining labels which are not equal to $p(z)$. 
%Those label the neighboring phase regions.
Now we can define
\begin{align}
    D_i(z) =
    \left\lbrace 
\begin{array}{ll}
    - \min( \delta_{i+1}(z),\delta_{i+2}(z)),  & i = p(z) \\
    \delta_{i}(z).
\end{array}
    \right.
\end{align}
where it is understood that $\delta_{3+k}(z) = \delta_{k}(z)$, for $0 \leq k \leq 2$.
% \begin{align}
%     D_{n_i(z)}(z) &= \delta_{n_i(z)}(z), 
%     \\
%     D_{p(z)}(z)& = - \min( D_{n_1(z)}(z),D_{n_2(z)}(z))
% \end{align}
Then we define sigmoids
\begin{equation}
    \sigma_i(z) = ( 1 + e^{D_i(z)/\ell} )^{-1},
\end{equation}
where $\ell$ controls the width of the transition region.
If $z$ is contained deep in phase region $k$ we are guaranteed that $ \sigma_k(z) \approx 1$, while $ \sigma_{i}(z) \approx 0$ for the other phases $i \neq k$.
We are finally ready to define
\begin{align}
    \chi_i(z) =  \frac{\sigma_i( z )}{\sigma_1( z )+\sigma_2(z) +\sigma_3(z)} .
\end{align}
In figure Fig.(\ref{fig:Y_junction_design}), we give an example of how the partitions look like for $\ell=0.1,0.25, 0.5$. The partition used in the main text is generated with $\ell=0.1$.

\begin{figure}[hbt]
    \centering
\includegraphics[width=\linewidth]{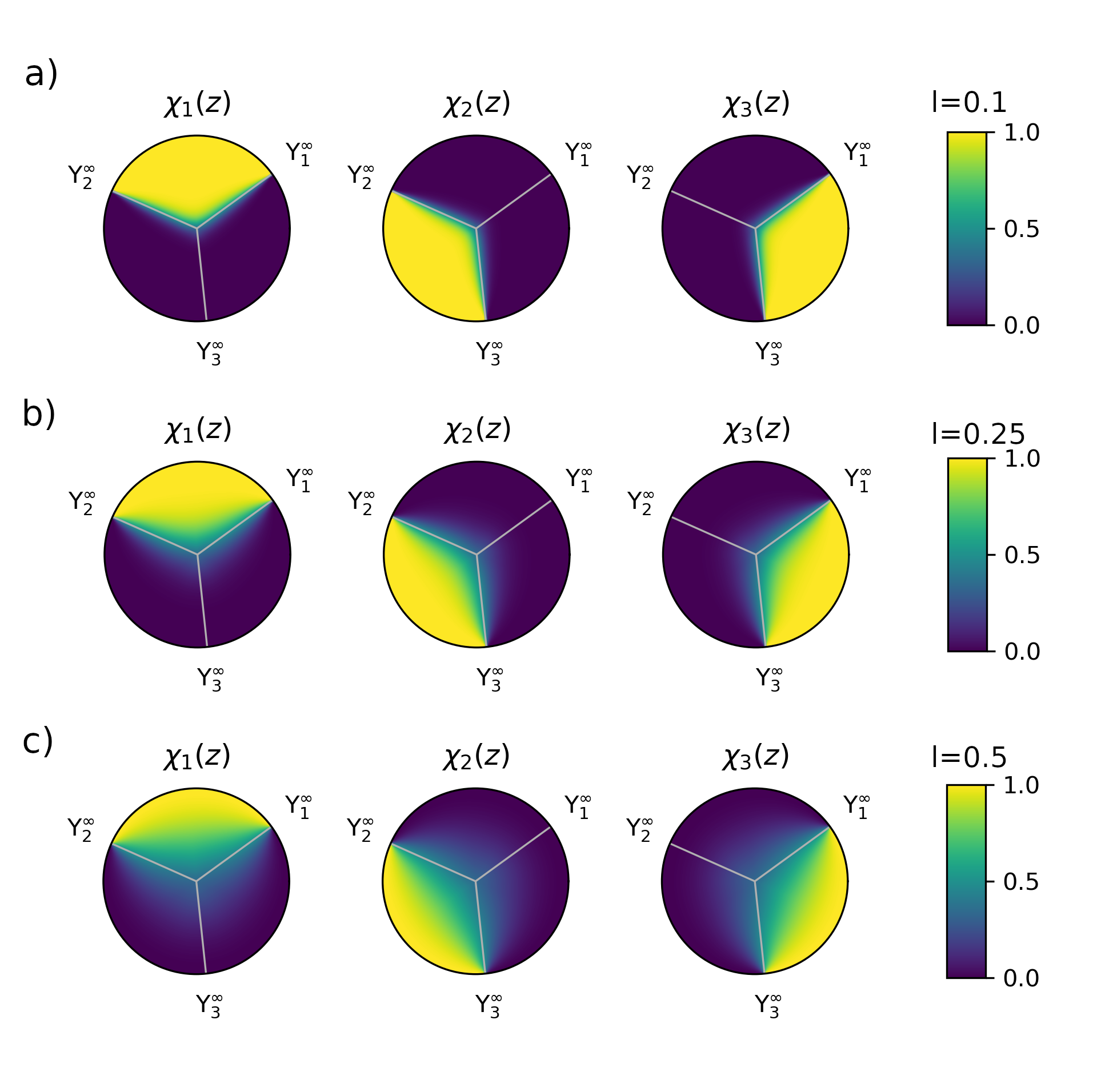}
    \caption{{\bf Y-junction design}. We define three geodesic phase boundaries $Y_i$ which meet at $z=0$ and separate the hyperbolic disk into three phase regions. The indicator functions $\chi_i(z)$ signal the presence of phase $i$ and form a partition of unity $\sum_i \chi_i(z)=1$. The parameter $\ell=0.1,0.25,0.5$ controls the width of the smooth domain wall between the phases. }
    \label{fig:Y_junction_design}
\end{figure}

\subsection{Computational aspects}

The code for all computations presented in the paper can be downloaded from \cite{CodeRepository}.
The spectral densities shown in Fig. 3 and Fig. 4 were created using the kernel polynomial method with a Jackson kernel \cite{Weisse2006}.
500 spectral moments were calculated and the traces were evaluated using 10 random states.
The computations for Fig. 2 rely on the exact diagonalization of the adjacency operator using OpenBLAS.
\clearpage


\begin{thebibliography}{1}

\bibitem{RuzzeneEM2021} M. Ruzzene, E. Prodan, C. Prodan, {\sl Dynamics of elastic hyperbolic lattices}, Extreme Mechanics Letters {\bf 49}, 101491 (2021).

\bibitem{KollarNature2019} A. J. Koll\'ar, M. Fitzpatrick, A. A. Houck, {\sl Hyperbolic lattices in circuit quantum electrodynamics}, Nature {\bf 571}, 45-50 (2019).

\bibitem{LenggenhagerNatComm2022} P. M. Lenggenhager {\it et al}, {\sl Simulating hyperbolic space on a circuit board}, Nature Comm. {\bf 13}, 4373 (2022).

\bibitem{ZhangNatComm2022} W. Zhang, H. Yuan, N. Sun, H. Sun, X. Zhang, {\sl Observation of novel topological states in hyperbolic lattices}, Nature Comm. {\bf 13}, 2937 (2022).

\bibitem{ChenNatComm2023} A. Chen, H. Brand, T. Helbig, T. Hofmann, S. Imhof, A. Fritzsche, T. Kiebling, A. Stegmaier, L. K. Upreti, T. Neupert, T. Bzdušek, M. Greiter, R. Thomale, I. Boettcher, {\sl Hyperbolic matter in electrical circuits with tunable complex phases}, Nature Communications {\bf 14}, 622 (2023).

\bibitem{ZhangNatComm2023} W. Zhang, F. Di, X. Zheng, H. Sun, X. Zhang, {\sl Hyperbolic band topology with non-trivial
second Chern numbers}, Nature Communications {\bf 14}, 1083 (2023).

\bibitem{UrwylerPRL2022} D. M. Urwyler, P. M. Lenggenhager, I. Boettcher, R. Thomale, T. Neupert, T. Bzdušek, {\sl Hyperbolic topological band insulators}, Phys. Rev. Lett. {\bf 129}, 246402 (2022).

\bibitem{LiuPRB2023} Z.-R. Liu, C.-B. Hua, T. Peng, R. Chen, B. Zhou, {\sl Higher-order topological insulators in hyperbolic lattices}, Phys. Rev. B {\bf 107}, 125302 (2023).


\bibitem{MathaiATMP2019} V. Mathai, G. C. Thiang, {\sl Topological phases on the hyperbolic plane: fractional bulk-boundary correspondence}, Adv. Theor. Math. Phys. {\bf 23}, 803-840 (2019).

\bibitem{KollarCMP2020} A. J. Koll\'ar, M. Fitzpatrick, P. Sarnak, A. A. Houck, {\sl Line-graph lattices: Euclidean and non-Euclidean flat bands, and implementations in circuit quantum electrodynamics}, Comm. Math. Phys. {\bf 376}, 1909-1956 (2020).

\bibitem{ComtetAP1987} A. Comtet, {\sl On the Landau levels on the hyperbolic plane}, Annals of Physics {\bf 173}, 185-209 (1987).

\bibitem{CareyCMP1998} A. L. Carey, K. C. Hannabuss, V. Mathai, P. McCann, {\sl Quantum Hall effect on the hyperbolic plane}, Commun. Math. Phys. 190, 629-673 (1998).

\bibitem{MarcolliCCM1999} M. Marcolli and V. Mathai, {\sl Twisted index theory on good orbifolds, I: Noncommutative Bloch theory}, Communications in Contemporary Mathematics {\bf 1}, 553-587 (1999).

\bibitem{MarcolliCMP2001} M. Marcolli, V. Mathai, {\sl Twisted index theory on good orbifolds, II: Fractional quantum numbers}, Commun. Math. Phys. {\bf 217}, 55-87 (2001).

\bibitem{LuekKTh2000} W. L\"uck, R. Stamm, {\sl Computations of K- and L-theory of cocompact planar groups}, K-Theory {\bf 21}, 249-292 (2000).

\bibitem{MassattMMS2017} D. Massatt, M. Luskin, C. Ortner, {Electronic density of states for incommensurate layers}, Multiscale Modeling \& Simulation {\bf 15}, 476-499 (2017). 

\bibitem{ProdanSpringer2017} E. Prodan, {\sl A computational non-commutative geometry program for disordered topological insulators}, (Springer, Berlin, 2017).

\bibitem{LuxArxiv2023} F. R. Lux, E. Prodan, {\sl Spectral and combinatorial aspects of Cayley-crystals}, arXiv:2212.10329 (2022).

\bibitem{MaciejkoSciAdv2021} J. Maciejko, S. Rayan, {\sl Hyperbolic band theory}, Sci. Adv. {\bf 7}, eabe9170 (2021).

\bibitem{ChengPRL2022} N. Cheng, F. Serafin, J. McInerney, Z. Rocklin, K. Sun, X. Mao, {\sl Theory and boundary modes of high-dimensional representations of infinite hyperbolic lattices}, Phys. Rev. Lett. {\bf 129}, 088002 (2022).

\bibitem{MaciejkoPNAS2022} J. Maciejko, S. Rayan, {\sl Automorphic Bloch theorems for hyperbolic lattices}, Proc. Nat. Acad. Sci. {\bf 119}, e2116869119 (2022).

%\bibitem{ProdanSpringer2016} E. Prodan, H. Schulz-Baldes, {\sl Bulk and boundary invariants for complex topological insulators: From K-theory to physics}, (Mathematical Physics Studies, Springer, 2016).

\bibitem{LuekGFA1994} W. L\"uck, {\sl Approximating $L^2$-invariants by their finite-dimensional analogues}, Geom. Funct. Anal. {\bf 4}, 455–481 (1994).

\bibitem{LuekBook2002} W. L\"uck, {\sl $L^2$-invariants: Theory and applications to geometry and K-theory}, Ergeb. Math. Grenzgeb. (3), vol. 44, Springer-Verlag, Berlin, 2002.

\bibitem{SaussetJPA2007} F. Sausset, G. Tarjus, {\sl Periodic boundary conditions on the pseudosphere}, J. Phys. A: Math. Theor. {\bf 40}, 12873-12899 (2007).


\bibitem{Footnote1} The band carrying the hyperbolic Chern number will be resolved in a separate study. 

\bibitem{Katok1992} S. Katok, {\sl Fuchsian groups}, (Univ. of Chicago Press, Chicago, 1992).

\bibitem{BoettcherPRB2022} I. Boettcher, A. V. Gorshkov, A. J. Koll\'ar, J. Maciejko, S. Rayan, R. Thomale, {\sl Crystallography of hyperbolic lattices}, Phys. Rev. B {\bf 105}, 125118 (2022).

\bibitem{Supplemental} Supplemental Material.

\bibitem{YunckenMMJ2003} R. Yuncken, {\sl Regular tessellations of the hyperbolic plane by fundamental domains of a Fuchsian group},  Moscow Mathematical Journal {\bf 3}, 4 (2003).

\bibitem{CoxeterBook} H. S. M. Coxeter . W. O. J. Moser, {\sl Generators and relations for discrete groups}, (Springer, Berlin, 1972).

\bibitem{ApigoPRM2018} D. J. Apigo, K. Qian, C. Prodan, E. Prodan, {\sl Topological edge modes by smart patterning}, Phys. Rev. Materials {\bf 2}, 124203 (2018).

\bibitem{HumphreysBook1990} J. E. Humphreys, {\sl Reflection groups and Coxeter groups}, (Cambridge U. Press, Cambridge, 1990).

\bibitem{Mennicke1967} J. Mennicke, {\sl Eine Bemerkung \"uber Fuchssche Gruppen}, Invent. Math. {\bf 2}, 301-305 (1967)

\bibitem{Siran2001a} J. \v{S}ir\`a\v{n}, {\sl The triangle group representations and their applications to graphs and maps}, Discrete Math. {\bf 229}, 341-358 (2001)

\bibitem{Siran2001b} J. \v{S}ir\`a\v{n}, {\sl The triangle group representations and constructions of regular maps}, Proc. London Math. Soc. {\bf 82}, 513-532 (2001)

\bibitem{Lehmer1933} D. Lehmer, {\sl A note on trigonometric algebraic numbers}, Amer. Math. Mon. {\bf 40}, 165-166 (1933)

\bibitem{Watkins1993} W. Watkins, J. Zeitlin, {\sl The minimal polynomial of $\cos(2\pi/n)$}, Amer. Math. Mon. {\bf 100}, 471-474 (1993)

\bibitem{CodeRepository} \url{https://github.com/luxfabian/hyperbolic_crystals}

\bibitem{footnote7} Degrees of freedom can be turned on or off during these continuous deformations.

\bibitem{footnote9} Model Hamiltonians delivering topological bands with $n_0=1$ appeared already in \cite{UrwylerPRL2022}.

\bibitem{Wang2019} G. Wang, M. Vuorinen, X. Xhang, {\sl On cyclic quadrilaterals in euclidean and hyperbolic geometries}, Publ. Math. Debrecen  {\bf 99}, 123–140 (2021).


\bibitem{Weisse2006} A. Wei\ss e,  G. Wellein, A. Alvermann, H. Fehske, {\sl The kernel polynomial method}, Rev. Mod. Phys. {\bf 78}, 275 (2006)

%arXiv:1908.10389 [math.MG] (2019)

%\bibitem{DupuyJMP} A. Aigon-Dupuy, P. Buser, M. Cibils, A. F. K\"unzle, {\sl Hyperbolic octagons and Teichm\"uller space in genus 2}, J. Math. Phys. {\bf 46}, 033513 (2005).

%\bibitem{} P. Bienias, I. Boettcher, R. Belyansky, A. J. Kollar, A. V., Gorshkov, {\sl Circuit Quantum Electrodynamics in Hyperbolic Space: From Photon Bound States to Frustrated Spin Models}, Phys. Rev. Lett. {\bf 128}, 013601 (2022).


\end{thebibliography}
\end{document}